\documentclass[runningheads,envcountsame, envcountsect
]{llncs}
\usepackage{amsmath, nccmath}
\usepackage{amsfonts, amssymb, mathtools}
\usepackage{color, ifthen} %
\usepackage{graphicx}
\usepackage{proof}
\usepackage{cite}
\usepackage{booktabs}
\usepackage{siunitx, multirow}
\usepackage[table]{xcolor}
\usepackage[all]{xy} \SelectTips{cm}{}  %
\usepackage{url}
\usepackage{subcaption}
\usepackage{listings}
\usepackage{lsthaskell}
\usepackage{hyperref}

\usepackage{wrapfig}
\usepackage{cutwin, lipsum}
\opencutright
\usepackage{algorithm2e}
\renewcommand{\triangleq}{\mathbin{:=}}

\newcommand{\rloop}[2][-]{\save \POS!R(.7) \ar@(ru,rd)^#1{#2} \restore}
\newcommand{\lloop}[2][-]{\save \POS!L(.7) \ar@(lu,ld)_#1{#2} \restore}
\newcommand{\uloop}[2][-]{\save \POS!U(.7) \ar@(lu,ru)^(.8){#2} \restore}

\newcommand{\Pred}{\mathbf{Pred}}
\newcommand{\congrightarrow}{\mathrel{\stackrel{
           \raisebox{.5ex}{$\scriptstyle\cong\,$}}{
           \raisebox{0ex}[0ex][0ex]{$\rightarrow$}}}}
\newcommand{\iso}{\congrightarrow}

\newcommand{\relto}{-{\kern-1.5ex}\raisebox{-.1pt}{\mbox{$\shortmid$}}{\kern-2ex}\to} %

\newcommand{\Set}{\mathbf{Set}}
\newcommand{\Pf}{\mathcal{P}}
\newcommand{\dotP}{\dot{\mathcal{P}}}
\newcommand{\D}{\mathcal{D}}

\newcommand{\CLatw}[0]{\textbf{CLat}$_{\wedge}$-}

\newcommand{\supp}[1]{\mathrm{supp}(#1)}
\newcommand{\op}{\mathrm{op}}
\newcommand{\Act}{\mathrm{Act}}

\newcommand{\KTomega}{$\text{KT}^{\omega}$}

\newcommand{\Pref}[1]{\mathbf{Pre}(#1)}
\newcommand{\Postf}[1]{\mathbf{Post}(#1)}

\newcommand{\backwardGFP}[0]{backward safety problem\xspace}
\newcommand{\forwardLFP}[0]{forward safety problem\xspace}
\newcommand{\backwardLFP}[0]{inverse backward safety problem\xspace}

\newcommand{\Gmdp }{G}
\newcommand{\Fmdp}{F'}
\newcommand{\Fmrm}{F'}
\newcommand{\Fkrf}{F''}
\newcommand{\Fkrb}{F'''}

\newcommand{\myparagraph}[1]{\vspace*{.3em}\noindent\textbf{#1}\;}
\newcommand{\FKripke}{$\textbf{PDR}^{\textbf{F-Kr}}$}
\newcommand{\IBKripke}{$\textbf{PDR}^{\textbf{IB-Kr}}$}
\newcommand{\IBMDP}{$\textbf{PDR}^{\textbf{IB-MDP}}$}
\newcommand{\MRM}{$\textbf{PDR}^{\textbf{MRM}}$}

\usepackage{comment}
\specialcomment{auxproof}
{\mbox{}\newline\textbf{BEGIN: AUX-PROOF}\dotfill\newline}
{\mbox{}\newline\textbf{END: AUX-PROOF}\dotfill\newline}

\newif\ifdraft\draftfalse

\newif\ifarxiv\arxivtrue

\ifdraft
\usepackage{todonotes}
\newcommand\kori[1]{\textcolor{blue}{#1}}
\newcommand\korit[1]{\todo[color=green!40]{#1 --kori}}
\newcommand{\conf}[1]{}
\newcommand{\todoil}[1]{\todo[inline,caption={}]{#1}}
\else
\usepackage[disable]{todonotes}
\newcommand\kori[1]{#1}
\newcommand\korit[1]{}
\excludecomment{auxproof}
\newcommand{\conf}[1]{}
\newcommand{\todoil}[1]{}
\fi

\ifarxiv
\newcommand\finalarxiv[2]{#2}
\else
\newcommand\finalarxiv[2]{#1}
\fi

\spnewtheorem{notation}[theorem]{Notation}{\bfseries}{\upshape}

\begin{document}
\title{The Lattice-Theoretic Essence of Property Directed Reachability
  Analysis
\thanks{
  The authors are supported by ERATO HASUO Metamathematics for Systems Design Project (No.~JPMJER1603).
  MK is a JSPS DC fellow and supported by JSPS KAKENHI Grant (No.~22J21742).
  KS is supported by JST CREST Grant (No.~JPMJCR2012) and JSPS KAKENHI Grant (No.~19H04084).
  }}

\author{ Mayuko Kori\inst{1, 2}\orcidID{0000-0002-8495-5925} \and
  Natsuki Urabe \inst{2} \orcidID{0000-0002-1554-6618} \and Shin-ya
  Katsumata \inst{2}\orcidID{0000-0001-7529-5489} \and Kohei Suenaga
  \inst{3}\orcidID{0000-0002-7466-8789} \and Ichiro
  Hasuo\inst{1,2}\orcidID{0000-0002-8300-4650} }
\authorrunning{M. Kori et al.}
\institute{The Graduate University for Advanced Studies (SOKENDAI),
  Hayama, Japan \and
  National Institute of Informatics, Tokyo, Japan \\
  \email{\{mkori, urabenatsuki, s-katsumata, hasuo\}@nii.ac.jp} \and
  Kyoto University, Kyoto, Japan \\
  \email{ksuenaga@fos.kuis.kyoto-u.ac.jp} }

\maketitle %
\begin{abstract}
  We present {\em LT-PDR}, a lattice-theoretic generalization of  Bradley's property directed reachability analysis (PDR) algorithm. LT-PDR identifies the essence of PDR to be an ingenious combination of verification and refutation attempts based on the Knaster--Tarski and Kleene theorems.
We introduce four  concrete instances of LT-PDR, derive their implementation from a generic Haskell implementation of LT-PDR, and experimentally evaluate them. We also present a categorical structural theory that derives these instances.

  \keywords{property directed reachability analysis \and model checking \and lattice theory \and
    fixed point theory \and category theory }
\end{abstract}

\section{Introduction}

{\em Property directed reachability (PDR)}  (also called \emph{IC3}) introduced in~\cite{Bradley11,EenMB11}
is a model checking algorithm for
proving/disproving safety problems. It has been successfully applied
to software and hardware model checking, and later it has been
extended in several directions, including {\em fbPDR}
\cite{SeufertS18,SeufertS19} that uses both forward and backward
predicate transformers and {\em PrIC3} \cite{BatzJKKMS20} for the
quantitative safety problem for probabilistic systems. See~\cite{Gurfinkel2015IC3PA} for a concise overview.

The original PDR  assumes that systems are given by binary
predicates representing  transition relations.   The PDR
algorithm maintains  data structures called {\em
  frames} and {\em proof obligations}---these are collections of
predicates over states---and updates them.
While this logic-based description immediately yields automated tools using SAT/SMT solvers, it limits
target systems to qualitative and
nondeterministic ones. This limitation was first overcome
by PrIC3 \cite{BatzJKKMS20} whose target is probabilistic systems. This suggests room for
further generalization of  PDR.

In this paper, we propose the first lattice theory-based generalization
of the PDR algorithm; we call it  \emph{LT-PDR}.
This makes the PDR algorithm apply to a wider class of safety problems, including qualitative and quantitative. We also derive a new concrete extension of PDR, namely one for Markov reward models.

We implemented the general algorithm LT-PDR in Haskell, in a way that maintains the theoretical abstraction and clarity. Deriving concrete instances for various types of systems is easy (for Kripke structures, probabilistic systems, etc.). We conducted an experimental evaluation, which shows that these easily-obtained instances have at least reasonable performance.

\myparagraph{Preview of the Theoretical Contribution}
We generalize the PDR algorithm so that it operates over an
  arbitrary complete lattice $L$.  This generalization recasts the PDR
  algorithm to solve a general problem  $\mu F\leq^{?}\alpha$ of
  over-approximating the least fixed point of an $\omega$-continuous function $F\colon L\to L$ by a
  safety property $\alpha$.  This lattice-theoretic generalization
  signifies the relationship between the PDR algorithm and the theory of
  fixed points. This also allows us to incorporate quantitative
  predicates suited for probabilistic verification.

 More specifically,
 we reconstruct the original PDR algorithm as a combination of
  two constituent parts. They are  called {\em positive LT-PDR} and {\em negative LT-PDR}.  Positive LT-PDR comes
  from a witness-based proof method by the \emph{Knaster--Tarski fixed point
  theorem}, and aims to \emph{verify} $\mu F\leq^{?}\alpha$.
  In contrast, negative LT-PDR comes from the \emph{Kleene fixed point theorem} and aims to \emph{refute}  $\mu F\leq^{?}\alpha$.
  The two algorithms build up witnesses in an iterative and
  nondeterministic manner, where nondeterminism
  accommodates guesses and heuristics. We identify the essence of PDR
  to be an ingenious combination of these two algorithms, in which
  intermediate results on one side (positive or negative) give
  informed guesses on the other side.
  This is how we formulate LT-PDR in~\S{}\ref{sec:int}.

 We discuss several instances of our general theory of PDR. We discuss three concrete settings:
 Kripke structures (where we obtain two instances of LT-PDR),
Markov decision
  processes (MDPs), and Markov reward models. The two in the first setting essentially subsume many existing PDR
  algorithms, such as the original PDR~\cite{Bradley11,EenMB11}  and Reverse PDR~\cite{SeufertS18,SeufertS19}, and  the  one for MDPs
  resembles PrIC3~\cite{BatzJKKMS20}. The last one (Markov reward models) is a new
  algorithm that fully exploits the generality of our
   framework.

In fact, there is another dimension of theoretical generalization: the derivation of the above concrete instances follows a \emph{structural theory of state-based dynamics and predicate transformers}. We formulate the structural theory in the language of \emph{category theory}~\cite{MacLane71,Awodey06}---using especially \emph{coalgebras}~\cite{Jacobs16coalgBook} and \emph{fibrations}~\cite{CLTT}---following works such as~\cite{HermidaJ98,SprungerKDH18,KoriHK21,BonchiKP18}. The structural theory tells us which  safety problems  arise under what conditions; it can therefore suggest that certain safety problems are unlikely to be formulatable, too.  The structural theory is important because it builds a mathematical order in the PDR literature, in which theoretical developments tend to be closely tied to implementation and thus theoretical essences are often not very explicit. For example, the theory is useful in classifying a plethora of PDR-like algorithms for Kripke structures (the original, Reverse PDR, fbPDR, etc.). See \S\ref{sec:LTPDRsForKripke}.

We present the above structural theory in \S\ref{sec:strTh} and briefly discuss its use in  the derivation of concrete instances in \S\ref{sec:instances}. We note, however, that  this categorical theory is not needed for reading and  using the other parts of the paper.

There are other works on generalization of PDR~\cite{HoderB12,RinetzkyS16}, but
our identification of the interplay of Knaster--Tarski and Kleene is new. They do not accommodate probabilistic verification, either.
See \finalarxiv{\cite[Appendix A]{ArxivFull}}{Appendix~\ref{appendix:relatedWorkOnGen}} for further discussions.

\myparagraph{Preliminaries}
Let $(L,\le)$ be a poset.   $(L,\le)^\op$ denotes the opposite poset
$(L,\ge)$. Note that if $(L,\leq)$ is a complete lattice then so is
$(L,\le)^\op$.
An $\omega$-chain (resp. $\omega^{op}$-chain) in $L$ is an
$\mathbb{N}$-indexed family of increasing (resp. decreasing) elements
in $L$.  A monotone function $F:L\to L$ is {\em $\omega$-continuous}
(resp. $\omega^{op}$-continuous) if $F$ preserves existing suprema of
$\omega$-chains (resp. infima of $\omega^\op$-chains).

\section{Fixed-points in Complete Lattices}
Let $(L, \leq)$ be a complete lattice and $F: L \to L$ be a monotone
function.  When we analyze
fixed points %
of $F$, pre/postfixed points play important
roles. %

\begin{definition}
  A \emph{prefixed point} of $F$ is an element $x \in L$ satisfying
  $Fx \leq x$.  A \emph{postfixed point} of $F$ is an element
  $x \in L$ satisfying $x \leq Fx$.  We write $\Pref F$ and $\Postf F$
  for the set of prefixed points and postfixed points of $F$,
  respectively.
\end{definition}

The following
results %
are central in fixed point theory.  They allow us to
under/over-approximate the least/greatest fixed points.

\begin{theorem} \label{thm:kt_cc}
  A monotone endofunction $F$ on a complete lattice $(L, \leq)$ has
  the least fixed point $\mu F$ and the greatest fixed point $\nu
  F$. Moreover,  %
  \begin{enumerate}
  \item\label{item:thm:kt_cc1} (Knaster--Tarski~\cite{Tarski55}) The
    set of fixed points forms a complete lattice.  Furthermore,
    $\mu F = \bigwedge \{x \in L \mid Fx \leq x\}$ and
    $\nu F = \bigvee \{x \in L \mid x \leq Fx\}$.
  \item\label{item:thm:kt_cc2} (Kleene, see e.g.~\cite{Baranga91}) If
    $F$ is $\omega$-continuous,
    $\mu F=\bigvee_{n \in \mathbb{N}} F^n\bot$.  Dually, if $F$ is
    $\omega^\op$-continuous,
    $\nu F = \bigwedge_{n \in \mathbb{N}} F^n\top$.
    \qed
  \end{enumerate}
\end{theorem}

Thm.~\ref{thm:kt_cc}.\ref{item:thm:kt_cc2} is known to hold for
arbitrary $\omega$-cpos (complete lattices are their special case).
A generalization of Thm.~\ref{thm:kt_cc}.\ref{item:thm:kt_cc2} is the
Cousot--Cousot characterization~\cite{Cousot79}, where $F$ is assumed
to be monotone (but not necessarily $\omega$-continuous) and we have
$\mu F=F^\kappa\bot$ for a sufficiently large, possibly transfinite,
ordinal $\kappa$. In this paper, for the algorithmic study of PDR, we
assume the $\omega$-continuity of $F$. Note that $\omega$-continuous
$F$ on a complete lattice is necessarily monotone.

We call the $\omega$-chain $\bot \leq F\bot \leq \cdots $ \emph{the
  initial chain of $F$} and the $\omega^\op$-chain
$\top \geq F\top \geq \cdots $ \emph{the final chain of $F$}.
These appear in Thm.~\ref{thm:kt_cc}.\ref{item:thm:kt_cc2}.

Thm.~\ref{thm:kt_cc}.\ref{item:thm:kt_cc1} and~\ref{thm:kt_cc}.\ref{item:thm:kt_cc2} yield the following witness
notions for \emph{proving} and \emph{disproving} $\mu F \leq\alpha$,
respectively.
 \begin{corollary} \label{cor:kt_kleene} Let $(L,\leq)$ be a complete
  lattice and $F:L\to L$ be $\omega$-continuous.
  \begin{enumerate}
  \item\label{item:cor:kt_kleene1} (KT) $\mu F \leq\alpha$ if and only
    if there is $x\in L$ such that $Fx \leq x\leq \alpha$.
  \item\label{item:cor:kt_kleene2} (Kleene) $\mu F \not\leq\alpha$ if
    and only if there is $n\in\mathbb{N}$ and $x \in L$ such that
    $x \leq F^n \bot$ and $x \not \leq\alpha$.  \qed
  \end{enumerate}
 \end{corollary}

By Cor.~\ref{cor:kt_kleene}.\ref{item:cor:kt_kleene1}, proving
$\mu F \leq\alpha$ is reduced to searching for $x\in L$ such that
$Fx \leq x\leq \alpha$. We call such $x$ a \emph{KT (positive)
  witness}.  In contrast, by
Cor.~\ref{cor:kt_kleene}.\ref{item:cor:kt_kleene2}, disproving
$\mu F \leq\alpha$ is reduced to searching for $n\in\mathbb{N}$ and
$x \in L$ such that $x \leq F^n \bot$ and $x \not \leq\alpha$. We call
such $x$ a \emph{Kleene (negative) witness}.
\begin{notation} %
We shall use lowercase
  (Roman and Greek) letters for elements of $L$ (such as
  $\alpha, x\in L$), and uppercase letters for (finite or infinite)
  sequences of $L$ (such as $X\in L^{*}$ or $L^{\omega}$).  The $i$-th (or $(i-j)$-th when subscripts are started from $j$)
  element of a sequence $X$ is designated by a subscript:
  $X_{i}\in L$.
\end{notation}

\section{Lattice-Theoretic Reconstruction of PDR }
\label{sec:seqcon}

Towards the LT-PDR algorithm, we first introduce two simpler
algorithms, called positive LT-PDR (\S{}\ref{sec:pos}) and negative
LT-PDR (\S{}\ref{sec:neg}).  The target problem of the LT-PDR
algorithm is the following:
 \begin{definition}[the LFP-OA
  problem $\mu F \leq^{?} \alpha$]\label{def:lfpOverapprox}
  Let $L$ be a complete lattice, $F: L \to L$ be $\omega$-continuous,
  and $\alpha \in L$.  The \emph{lfp over-approximation (LFP-OA)
    problem} asks if $\mu F \leq \alpha$ holds; the problem is denoted
  by $\mu F \leq^{?} \alpha$.
 \end{definition}

 \begin{example}\label{ex:forward}
  Consider a transition system, where $S$ be the set of states,
  $\iota \subseteq S$ be the set of initial states,
  $\delta: S \to \Pf S$ be the transition relation, and
  $\alpha \subseteq S$ be the set of safe states.  Then letting
  $L\coloneqq \Pf S$ and
  $F \coloneqq \iota \cup \bigcup_{s \in (-)}\delta(s)$, the lfp
  over-approximation problem $\mu F \leq^{?} \alpha$ is the problem
  whether all reachable states are safe.  It is equal to the problem
  studied by the conventional IC3/PDR~\cite{Bradley11,EenMB11}.
 \end{example}

Positive LT-PDR iteratively builds a KT witness in a bottom-up manner
that positively answers the LFP-OA problem, while negative LT-PDR
iteratively builds a Kleene witness for the same LFP-OA
problem. We shall present these two algorithms as clear reflections of
two proof principles
(Cor.~\ref{cor:kt_kleene}), each of which comes from the fundamental
Knaster--Tarski and Kleene theorems.

The two algorithms build up witnesses in an iterative and
nondeterministic manner. The nondeterminism is there for accommodating
guesses and heuristics. We identify the essence of PDR to be an
ingenious combination of these two algorithms, in which intermediate
results on one side (positive or negative) give informed guesses on
the other side. This way, each of the positive and negative algorithms
provides heuristics in resolving the nondeterminism in the execution of
the other. This is how we formulate the LT-PDR algorithm
in~\S{}\ref{sec:int}.

The dual of LFP-OA problem is called the \emph{gfp-under-approximation
  problem} (GFP-UA): the GFP-UA problem for a complete lattice $L$,
  an $\omega^\op$-continuous function $F:L\to L$  and
  $\alpha\in L$
  is whether the inequality $\alpha\leq \nu F$ holds or not,
and is denoted by $\alpha\le^{?}\nu F$.  It is evident that the GFP-UA
problem for $(L, F, \alpha)$ is equivalent to the LFP-OA problem for
$(L^\op,F,\alpha)$. This suggests the dual algorithm called LT-OpPDR
for GFP-UA problem. See Rem.~\ref{rem:LTOpPDR} later.

\subsection{Positive LT-PDR: Sequential   Positive Witnesses}\label{sec:pos}

We introduce the notion of KT$^\omega$ witness---a KT witness
(Cor.~\ref{cor:kt_kleene}) constructed in a sequential manner. Positive LT-PDR searches for a KT$^\omega$ witness by growing its finitary
approximations (called KT sequences).

Let $L$ be be a complete lattice. We regard each element $x\in L$ as
an abstract presentation of a predicate on states. The inequality $x\le y$ means
that the predicate $x$ is stronger than the predicate $y$.  We introduce
the complete lattice $[n, L]$ of increasing chains of length
$n\in\mathbb{N}$, whose elements are $(X_0 \leq \cdots \leq X_{n-1})$
in $L$ equipped with the element-wise order. We similarly introduce
the complete lattice $[\omega, L]$ of $\omega$-chains in $L$.  We lift
$F:L\to L$ to $F^\# : [\omega, L] \to [\omega, L]$ and
$F^\#_n: [n, L] \to [n, L]$ (for $n \geq 2$) as follows. Note that the
entries are shifted.
\begin{equation}
\begin{aligned}
   F^\#(X_0 \leq X_1 \leq \cdots)&\;:=\; (\bot \leq FX_0 \leq FX_1 \leq\cdots)\\[-.3em]
  F^\#_n(X_0 \leq \cdots \leq X_{n-1})\;&:=\; (\bot \leq FX_0 \leq \cdots\leq FX_{n-2})
\end{aligned}
\end{equation}

 \begin{definition}[\KTomega{} witness]\label{def:KTseqwitness}
	Let $L,F,\alpha$ be as in  Def.~\ref{def:lfpOverapprox}.
	Define
	$\Delta \alpha := (\alpha \leq \alpha \leq \cdots)$.
	A \emph{\KTomega{} witness} is $X \in [\omega, L]$ such that $F^\# X \leq X \leq \Delta \alpha$.
 \end{definition}

 \begin{theorem} \label{thm:safe_witness}
  Let $L,F,\alpha$ be as in  Def.~\ref{def:lfpOverapprox}.
  There exists a KT  witness (Cor.~\ref{cor:kt_kleene})
  if and only if there exists a \KTomega  witness.
  \qed
 \end{theorem}

Concretely, a KT witness $x$ yields a \KTomega\ witness $x\le x\le\cdots$; a \KTomega\ witness $X$ yields a KT witness $\bigvee_{n \in \omega} X_{n}$.
A full proof (via Galois connections) is in
\finalarxiv{\cite{ArxivFull}}{the appendix}.

The initial chain $\bot \le F\bot \le \cdots$ is always a \KTomega\ witness for $\mu F \leq \alpha$.
There are other  \KTomega\  witnesses whose growth is accelerated by some heuristic guesses---an extreme example is $x\le x\le\cdots$ with a KT witness $x$.
\KTomega\ witnesses embrace the spectrum of such different sequential witnesses for $\mu F\le \alpha$, those which mix
routine constructions (i.e.\ application of $F$) and heuristic guesses.
\begin{definition}[KT sequence] \label{def:kt_sequence}
  Let $L,F,\alpha$ be as in  Def.~\ref{def:lfpOverapprox}.
A \emph{KT sequence} for
\todo{check "finite chain"}
 $\mu F \leq^{?} \alpha$ is a finite chain
  $(X_0\le\cdots\le X_{n-1})$, for $n\geq 2$, satisfying
  \begin{enumerate}
  \item \label{item:xn} $X_{n-2} \leq \alpha$; and
  \item \label{item:fn-alg} $X$ is a prefixed point of $F^\#_n$, that is, $FX_{i}\le X_{i+1}$ for each $i\in [0, n-2]$.
  \end{enumerate}
  A KT sequence $(X_{0} \leq \cdots \leq X_{n-1})$ is
  \emph{conclusive} if $X_{j+1} \le X_{j}$ for some $j$.
\end{definition}
 KT sequences are finite by definition. Note  that the upper bound $\alpha$ is imposed on all $X_{i}$ but $X_{n-1}$. This freedom in the choice of $X_{n-1}$ offers room for heuristics, one that is exploited in the combination with negative LT-PDR (\S{}\ref{sec:int}).

We take KT sequences as finite approximations of \KTomega\ witnesses.
This view shall be justified by the partial order
$(\preceq)$ between KT sequences defined below.
 \begin{definition}[order  $\preceq$ between KT sequences]
 We define a partial order relation $\preceq$ on  KT sequences
 as follows: $(X_0, \dots, X_{n-1}) \preceq (X'_0, \dots, X'_{m-1})$ if
 $n \leq m$ and $X_j \geq X_j'$ for each $0 \leq j \leq n-1$.
 \end{definition}

The order $X_j\geq X_j'$ represents that $X_j'$ is
a stronger predicate (on states) than $X_j$. Therefore $X\preceq X'$ expresses
that $X'$ is a longer and stronger / more determined chain than $X$. We obtain \KTomega\
witnesses as their $\omega$-superma.

  \begin{theorem}\label{thm:KTseqCPO}
    Let $L,F,\alpha$ be as in Def.~\ref{def:lfpOverapprox}.  The set
    of KT sequences, augmented with the set of \KTomega\
    witnesses %
    $\{X \in [\omega, L] \mid F^\# X \leq X \leq \Delta \alpha \}$ and
    ordered by the natural extension of $\preceq$, is an $\omega$-cpo.
    In this $\omega$-cpo, each \KTomega witness $X$ is represented as
    the suprema of an $\omega$-chain of KT sequences, namely
    $X = \bigvee_{n \geq 2} X|_n$ where $X|_n \in [n, L]$ is the
    length $n$ prefix of $X$.  \qed
 \end{theorem}

  \begin{proposition} \label{prop:safe}
    Let $L,F,\alpha$ be as in  Def.~\ref{def:lfpOverapprox}. There exists a \KTomega\ witness
  if and only if there exists a conclusive KT sequence.
 \end{proposition}

 \begin{proof}
  ($\Rightarrow$): If there exists a \KTomega\ witness,
  $\mu F \leq \alpha$ holds by Cor.~\ref{cor:kt_kleene} and
  Thm.~\ref{thm:safe_witness}.  Therefore, the ``informed guess'' $(\mu F \leq \mu F)$ gives a
  conclusive KT sequence.  ($\Leftarrow$): When $X$ is a conclusive KT sequence with $X_j = X_{j+1}$,
  $X_0 \leq \cdots \leq X_j = X_{j+1} = \cdots$ is a \KTomega\ witness. \qed
 \end{proof}

\noindent
The proposition above yields the following partial algorithm that aims to answer positively to
the LFP-OA problem. %
It searches for a conclusive KT sequence.

\begin{definition}[positive LT-PDR]
  Let $L,F,\alpha$ be as in  Def.~\ref{def:lfpOverapprox}. \emph{Positive LT-PDR} is
 the algorithm shown in Alg.~\ref{alg:posi_pdr}, which says
  `True' to the LFP-OA problem $\mu F \leq^{?} \alpha$ if successful.
\end{definition}

\begin{figure}[p]
\begin{minipage}{\textwidth}
    \begin{algorithm}[H]
    \caption{positive LT-PDR}\label{alg:posi_pdr}
    \SetKwInOut{Input}{Input} \SetKwInOut{Output}{Output}
    \SetKwInOut{Initially}{Initially} \Input{An instance
      ($\mu F \leq^? \alpha$) of the LFP-OA problem in $L$}
       \SetKwRepeat{Repeat}{repeat (do
      one of the following)}{until}
    \Output{`True' with a conclusive KT sequence} \KwData{a KT
      sequence $X = (X_0 \leq \dots \leq X_{n-1})$} \Initially{$X\coloneqq(\bot \leq F \bot)$}
    \Repeat{any return value is obtained}{ \textbf{Valid}
      If $X_{j+1} \leq X_j$ for some $j < n-1$, return `True' with the conclusive KT sequence $X$. \\
      \textbf{Unfold} If $X_{n-1} \leq \alpha$,
      let $X\coloneqq(X_0 \leq \cdots \leq X_{n-1} \leq \top)$, appending $\top$ \\
      \textbf{Induction} If some $k \geq 2$ and $x\in L$ satisfy
      $X_{k} \not \leq x$ and $F(X_{k-1} \land x) \leq x$,
      let $X \coloneqq X[X_j := X_j \land x]_{2 \leq j \leq k}$. \\
    }
  \end{algorithm}
\end{minipage}

\vspace{.5em}
\begin{minipage}{\textwidth}
\begin{algorithm}[H]
    \caption{negative LT-PDR}\label{alg:nega_pdr}
    \SetKwInOut{Input}{Input} \SetKwInOut{Output}{Output}
    \SetKwInOut{Initially}{Initially} \SetKwRepeat{Repeat}{repeat (do
      one of the following)}{until} \Input{An instance
      ($\mu F \leq^? \alpha$) of the LFP-OA problem in $L$}
    \Output{`False' with a conclusive Kleene sequence} \KwData{a
      Kleene sequence $C = (C_0, \dots, C_{n-1})$}
    \Initially{$C\coloneqq()$} \Repeat{any return value is obtained}{
      \textbf{Candidate} Choose $x\in L$ such that
      $x \not \leq \alpha$,
      and let $C\coloneqq(x)$. \\
      \textbf{Model}
      If $C_0 = \bot$, return `False' with the conclusive Kleene sequence $C$. \\
      \textbf{Decide} If there exists $x$ such that $C_0 \leq Fx$,
      then let $C \coloneqq (x, C_0, \dots, C_{n-1})$. \\
    }
  \end{algorithm}
\end{minipage}

\vspace{.5em}
\begin{minipage}{\textwidth}
  \begin{algorithm}[H]
    \caption{LT-PDR}\label{alg:pdr}
    \SetKwInOut{Input}{Input} \SetKwInOut{Output}{Output}
    \SetKwInOut{Initially}{Initially}
     \SetKwRepeat{Repeat}{repeat (do
      one of the following)}{until}
    \Input{An instance
      ($\mu F \leq^? \alpha$) of the LFP-OA problem in $L$}
    \Output{`True' with a conclusive KT sequence, or `False' with a
      conclusive Kleene sequence} \KwData{$(X; C)$ where $X$ is a KT
      sequence $(X_0 \leq \cdots \leq X_{n-1})$, and $C$ is a Kleene
      sequence $(C_i, C_{i+1}, \dots, C_{n-1})$ ($C$ is empty if
      $n=i$).}  \Initially{$(X; C)\coloneqq(\bot \leq F \bot;\; ()\,)$}
    \Repeat{any return value is obtained}{ \textbf{Valid}
      If $X_{j+1} \leq X_j$ for some $j < n-1$, return `True' with the conclusive KT sequence $X$. \\
      \textbf{Unfold} If $X_{n-1} \leq \alpha$,
      let $(X; C)\coloneqq(X_0 \leq \cdots \leq X_{n-1} \leq \top; ())$. \\
      \textbf{Induction} If some $k \geq 2$ and $x\in L$ satisfy
      $X_{k} \not \leq x$ and $F(X_{k-1} \land x) \leq x$,
      let $(X; C) \coloneqq (X[X_j := X_j \land x]_{2 \leq j \leq k}; C)$. \\
      \textbf{Candidate} If $C=()$ and $X_{n-1} \not \leq \alpha$,
      choose
      $x\in L$ such that $x \leq X_{n-1}$ and $x \not \leq \alpha$,
      and let $(X;C)\coloneqq(X; (x))$. \\
      \textbf{Model}
      If $C_1$ is defined, return `False' with the conclusive Kleene sequence $(\bot, C_1, \dots, C_{n-1})$. \\
      \textbf{Decide} If $C_i \leq FX_{i-1}$, choose $x \in L$
      satisfying $x \leq X_{i-1}$ and $C_i \leq Fx$,
      and let $(X; C) \coloneqq (X; (x, C_i, \dots, C_{n-1}))$. \\
      \textbf{Conflict} If $C_i \not \leq FX_{i-1}$, choose $x \in L$
      satisfying $C_i \not \leq x$ and $F(X_{i-1} \land x) \leq x$, and let
      $(X; C) \coloneqq (X[X_j := X_j \land x]_{2 \leq j \leq i};
      (C_{i+1}, \dots, C_{n-1}))$.  }
  \end{algorithm}
\end{minipage}
\end{figure}

The rules are designed by the following principles.

  \textbf{Valid} is applied when the current $X$ is
  conclusive.

 \textbf{Unfold}  extends $X$ with $\top$.  In fact, we
  can use any element $x$ satisfying $X_{n-1} \leq x$ and $FX_{n-1} \leq x$ in place of
  $\top$ (by the application of \textbf{Induction} with $x$).  The condition $X_{n-1} \leq \alpha$ is checked to ensure
  that the extended $X$ satisfies the condition in
  Def.~\ref{def:kt_sequence}.\ref{item:xn}.

 \textbf{Induction} strengthens $X$, replacing the $j$-th element with its meet with $x$.
  The first condition $X_{k} \not \leq x$ ensures that this rule indeed strengthens $X$,
  and the second condition
  $F(X_{k-1} \land x) \leq x$ ensures that the strengthened $X$ satisfies
  the condition in Def.~\ref{def:kt_sequence}.\ref{item:fn-alg}, that is,
  $F^\#_n X \leq X$ (see the proof in \finalarxiv{\cite{ArxivFull}}{Appendix~\ref{ap:config}}).

\begin{theorem}\label{thm:positive_sound_terminate}
 Let $L,F,\alpha$ be as in  Def.~\ref{def:lfpOverapprox}.
 Then positive LT-PDR is sound, i.e.~if it outputs `True' then $\mu F \leq \alpha$ holds.

 Moreover, assume $\mu F\le \alpha$ is true. Then positive LT-PDR is weakly
    terminating
    (meaning that suitable choices \kori{of $x$ when applying \textbf{Induction}} make the algorithm terminate).
    \qed
\end{theorem}

The last ``optimistic termination'' is realized by the informed guess
$\mu F$ as $x$ in \textbf{Induction}.
To guarantee the termination of LT-PDR,
it suffices to assume that the complete lattice $L$ is well-founded
(no infinite decreasing chain exists in $L$)
\kori{and there is no strictly increasing $\omega$-chain under $\alpha$ in $L$,}
although we cannot hope for this assumption in every instance (\S{} \ref{sec:LTPDRsForMDP},
\ref{sec:LTPDRsForMRM}).

  \begin{lemma} \label{lem:kt_order}
    Let $L,F,\alpha$ be as in  Def.~\ref{def:lfpOverapprox}.
 If $\mu F \leq \alpha$, then for any KT sequence $X$, at least one of the three rules in Algorithm~\ref{alg:posi_pdr} is enabled.

 Moreover, for any KT sequence $X$, let $X'$ be obtained by applying either \textbf{Unfold} or \textbf{Induction}. Then $X\preceq X'$ and  $X\neq X'$.
 \qed
 \end{lemma}

  \begin{theorem} \label{thm:posi_kt}
    Let $L,F,\alpha$ be as in  Def.~\ref{def:lfpOverapprox}.
  Assume that $\le$ in $L$ is well-founded and $\mu F \leq \alpha$.
 Then, any non-terminating run of positive LT-PDR converges to a \KTomega\ witness
 \kori{(meaning that it gives a \KTomega\ witness in $\omega$-steps)}.
 Moreover, if there is no strictly increasing $\omega$-chain bounded by $\alpha$ in $L$, then positive LT-PDR is strongly terminating.
 \qed
 \end{theorem}

\subsection{Negative PDR: Sequential Negative Witnesses
  \conf{70}}\label{sec:neg}
We next introduce \emph{Kleene sequences} as a
lattice-theoretic counterpart of
\emph{proof obligations} in the standard PDR. Kleene sequences
represent a chain of sufficient conditions to conclude that certain unsafe  states are reachable.

\begin{definition}[Kleene sequence] \label{def:kleene_sequence}
  Let $L,F,\alpha$ be as in  Def.~\ref{def:lfpOverapprox}. A \emph{Kleene sequence}
  for the LFP-OA problem $\mu F \leq^? \alpha$ is a finite sequence
  $(C_0, \dots, C_{n-1})$, for $n \geq 0$ ($C$ is empty if $n=0$), satisfying
  \begin{enumerate}
  \item \label{item:cfc} $C_j \leq FC_{j-1}$ for each
    $1 \leq j \leq n-1$;
  \item \label{item:cn} $C_{n-1} \not \leq \alpha$.
  \end{enumerate}
  A Kleene sequence $(C_0, \dots, C_{n-1})$ is \emph{conclusive} if
  $C_0 = \bot$.
  We may use $i \ (0 \leq i \leq n)$ instead of $0$ as the starting
  index of the Kleene sequence $C$.
 \end{definition}
When we have a Kleene sequence $C=(C_{0},\dots,C_{n-1})$, the
chain of implications $(C_{j}\le F^{j}\bot) \implies (C_{j+1}\le F^{j+1}\bot)$ hold for $0 \leq j < n-1$.
Therefore when $C$ is conclusive, $C_{n-1}$ is a
Kleene witness (Cor.~\ref{cor:kt_kleene}.\ref{item:cor:kt_kleene2}).

\begin{proposition} \label{prop:unsafe}
  Let $L,F,\alpha$ be as in  Def.~\ref{def:lfpOverapprox}.
   There exists a Kleene (negative)
  witness if and only if there exists a conclusive Kleene sequence.
 \end{proposition}
 \begin{proof}
  ($\Rightarrow$): If there exists a Kleene witness $x$ such that
  $x \leq F^n \bot$ and $x \not \leq \alpha$,
  $(\bot, F\bot, \dots, F^n \bot)$ is a conclusive Kleene sequence.
  ($\Leftarrow$): Assume there exists a conclusive Kleene sequence
  $C$. Then $C_{n-1}$ satisfies $C_{n-1} \leq F^{n-1}\bot$ and
  $C_{n-1} \not \leq \alpha$ because of
  $C_{n-1} \leq FC_{n-2} \leq \cdots \leq F^{n-1}C_0 = F^{n-1}\bot$ and Def.~\ref{def:kleene_sequence}.\ref{item:cn}.
  \qed
 \end{proof}

This proposition suggests the following algorithm to negatively answer to
the LFP-OA problem. It searches for a conclusive
Kleene sequence. The algorithm updates a Kleene sequence until its
first component becomes $\bot$.
\begin{definition}[negative LT-PDR]
  Let $L,F,\alpha$ be as in  Def.~\ref{def:lfpOverapprox}.
  \emph{Negative LT-PDR} is
  the algorithm shown in Alg.~\ref{alg:nega_pdr},
  which says `False' to the LFP-OA problem $\mu F \leq^? \alpha$ if successful.
\end{definition}
The rules are designed by the following principles.

 \textbf{Candidate} initializes $C$ with only one
  element $x$. The element $x$ has to be chosen such that
  $x \not \leq \alpha$ to ensure
  Def.~\ref{def:kleene_sequence}.\ref{item:cn}.

 \textbf{Model} is applied when the current Kleene
  sequence $C$ is conclusive.

 \textbf{Decide} prepends $x$ to $C$.
 The condition $C_0 \leq Fx$ ensures
  Def.~\ref{def:kleene_sequence}.\ref{item:cfc}.

\begin{theorem} \label{thm:negative}
  Let $L,F,\alpha$ be as in  Def.~\ref{def:lfpOverapprox}.
  \begin{enumerate}
  \item Negative LT-PDR is sound, i.e.\ if it outputs `False' then $\mu F \not\leq \alpha$.
  \item Assume $\mu F \not \le \alpha$ is true. Then negative LT-PDR
    is weakly terminating
    (meaning that suitable choices \kori{of $x$ when applying rules \textbf{Candidate} and \textbf{Decide}} make the algorithm terminate).
    \qed
  \end{enumerate}
\end{theorem}

\subsection{LT-PDR: Integrating Positive and Negative}\label{sec:int}
We have introduced two simple PDR algorithms, called positive LT-PDR
(\S{}\ref{sec:pos}) and negative LT-PDR (\S{}\ref{sec:neg}).  They are
so simple that they have potential inefficiencies. Specifically, in positive LT-PDR, it
is unclear that how we choose $x \in L$ in \textbf{Induction}, while
in negative LT-PDR, it may \kori{easily} diverge because the rules
\textbf{Candidate} and \textbf{Decide} may choose $x\in L$ that would
not lead to a conclusive Kleene sequence.  We resolve these inefficiencies by
combining positive LT-PDR and negative LT-PDR.  The combined PDR
algorithm is called LT-PDR, and it is a lattice-theoretic
generalization of conventional PDR.

Note that negative LT-PDR is only weakly terminating. Even worse, it is easy to make it diverge---after a choice of $x$ in \textbf{Candidate} or \textbf{Decide} such that $x\not\le \mu F$, no continued execution of the algorithm can lead to a conclusive Kleene sequence. For deciding $\mu F \leq^? \alpha$  efficiently, therefore, it is crucial to detect such useless Kleene sequences.

The core fact that underlies the efficiency of PDR is the following proposition, which says that a KT sequence (in positive LT-PDR) can quickly tell that a Kleene sequence (in negative LT-PDR) is useless. This fact is crucially used for many rules in LT-PDR (Def.~\ref{def:lt-pdr}).
\begin{proposition} \label{prop:C_X} Let $C=(C_i, \dots, C_{n-1})$ be a
  Kleene sequence $(2 \leq n, 0 < i \leq n-1)$ and
  $X=(X_0 \leq \cdots \leq X_{n-1})$ be a KT sequence. Then
  \begin{enumerate}
  \item \label{item:c_x} $C_i \not \leq X_i$ implies that $C$ cannot be extended to a
    conclusive one,
    that is, there does not exist
    $C_0, \dots, C_{i-1}$ such that $(C_0, \dots, C_{n-1})$ is
    conclusive.
  \item \label{item:c_fx} $C_i \not \leq F X_{i-1}$ implies that $C$
    cannot be extended to a conclusive one.
  \item \label{item:n-step} There is no conclusive Kleene sequence with length
    $n-1$.  \qed
  \end{enumerate}
\end{proposition}

The proof relies on the following lemmas.
\begin{lemma} \label{lem:x_init} Any KT sequence
  $(X_0 \leq \cdots \leq X_{n-1})$ over-approximates the initial
  sequence: $ F^i\bot\le X_i$ holds for any $i$ such that $0\le i\le n-1$.
  \qed
\end{lemma}

\begin{lemma} \label{lem:C_X}
  Let $C=(C_i, \dots, C_{n-1})$ be a
  Kleene sequence $(0 < i \leq n-1)$ and
  $(X_0 \leq \cdots \leq X_{n-1})$ be a KT sequence.  The following
   satisfy $1 \Leftrightarrow 2 \Rightarrow 3$.
  \begin{enumerate}
  \item \label{item:possible} The Kleene sequence $C$
    can be extended to a conclusive one.
  \item \label{item:fibot}$C_i \leq F^i\bot$.
  \item \label{item:fjx}$C_i \leq F^j X_{i-j}$ for each $j$ with $0 \leq j \leq i$.
  \qed
  \end{enumerate}
\end{lemma}

Using the above lattice-theoretic properties, we combine positive and
negative LT-PDRs into the following {\em LT-PDR} algorithm. It is
also a lattice-theoretic generalization of the original PDR
algorithm.  The combination
exploits the mutual
relationship between KT sequences and Kleene sequences,
exhibited as Prop.~\ref{prop:C_X}, for narrowing down choices in positive and negative LT-PDRs.
\begin{definition}[LT-PDR] \label{def:lt-pdr}
  Given a complete lattice $L$,
  an $\omega$-continuous function $F: L \to L$,
  and an element $\alpha \in L$,
  \emph{LT-PDR} is the algorithm
  shown in Alg.~\ref{alg:pdr} for the LFP-OA problem
  $\mu F \leq^? \alpha$.

\end{definition}
The rules are designed by the following principles.

(\textbf{Valid}, \textbf{Unfold}, and \textbf{Induction}): These rules
  are almost the same as in positive LT-PDR.
  In \textbf{Unfold},
  we reset the Kleene sequence because of
  Prop.~\ref{prop:C_X}.\ref{item:n-step}.
  Occurrences of \textbf{Unfold} punctuate an execution of the algorithm: between two occurrences of  \textbf{Unfold}, a main goal (towards a negative conclusion)  is to construct a conclusive Kleene sequence with the same length as the  $X$.

(\textbf{Candidate}, \textbf{Model}, and \textbf{Decide}): These rules
  have many similarities to those in negative LT-PDR. Differences are as follows:
  the \textbf{Candidate} and \textbf{Decide} rules impose $x \leq X_i$
  on the new element $x$ in $(x, C_{i+1}, \dots, C_{n-1})$ because
  Prop.~\ref{prop:C_X}.\ref{item:c_x} tells us that other choices are useless.  In \textbf{Model}, we only need to check
  whether $C_1$ is defined instead of $C_0 = \bot$. Indeed, since $C_1$ is
  added in \textbf{Candidate} or \textbf{Decide},
  $C_1 \leq X_1 = F\bot$ always holds.
  Therefore, $2 \Rightarrow 1$ in Lem.~\ref{lem:C_X} shows that $(\bot, C_1, \dots, C_{n-1})$ is
  conclusive.

(\textbf{Conflict}): This new rule emerges from the combination of
  positive and negative LT-PDRs.  This rule is applied when
  $C_i \not \leq FX_{i-1}$, which confirms that the current $C$ cannot
  be extended to a conclusive one (Prop.~\ref{prop:C_X}.\ref{item:c_fx}).
  Therefore, we eliminate $C_i$ from $C$ and strengthen $X$ so that we
  cannot choose $C_i$ again, that is, so that $C_i \not \leq (X_i \land x)$.
  Let us explain how $X$ is strengthened.  The element $x$ has to be
  chosen so that $C_i \not \leq x$ and $F(X_{i-1} \land x) \leq x$.
  The former dis-inequality ensures
  the strengthened $X$ satisfies $C_i \not \leq (X_i \land x)$,
  and
  the latter
  inequality implies $F(X_{i-1} \land x) \leq x$.
  One can see that \textbf{Conflict} is \textbf{Induction} with additional condition $C_i \not \leq x$, which enhances so that the search space for $x$ is narrowed down using the Kleene sequence $C$.

Canonical choices
of $x\in L$ in \textbf{Candidate}, \textbf{Decide}, and \textbf{Conflict}
are $x:=X_{n-1}$, $x:=X_{i-1}$, and $x:=FX_{i-1}$, respectively.
However, there can be cleverer choices; e.g.~$x:= S\setminus(C_i\setminus FX_{i-1})$ in \textbf{Conflict} when $L=\Pf S$.

 \begin{lemma} \label{lem:config} Each rule of LT-PDR, when applied to
  a pair of a KT  and a Kleene sequence, yields a pair of a
  KT and a Kleene sequence.  \qed
 \end{lemma}

 \begin{theorem}[correctness] \label{thm:pdrcor}
  LT-PDR is sound, i.e.~if it outputs `True'
  then $\mu F \leq \alpha$ holds, and if it outputs `False' then
  $\mu F \not \leq \alpha$ holds.  \qed
 \end{theorem}

Many existing PDR algorithms ensure termination if the state
space is finite.  \kori{A general principle behind is stated below. Note that it rarely applies to infinitary or quantitative settings, where we would need some abstraction for termination.}
 \begin{proposition}[termination] \label{prop:term}
	LT-PDR terminates regardless of the order of the rule-applications
	if the following conditions are satisfied.
	\begin{enumerate}
		\item\label{item:prop:term0}
		\textbf{Valid} and \textbf{Model} rules are
		immediately applied if applicable.
		\item\label{item:prop:term1}
		$(L,\leq)$ is well-founded.
    \item\label{item:prop:term2} Either of the following is
      satisfied: a) $\mu F \leq \alpha$ and $(L, \leq)$ has
      no strictly increasing $\omega$-chain bounded by $\alpha$,
      or b) $\mu F \not \leq \alpha$.  \qed
	\end{enumerate}
 \end{proposition}
Cond.~\ref{item:prop:term0} is natural: it just requires LT-PDR to immediately conclude `True' or `False' if it can.
Cond.~\ref{item:prop:term1}--\ref{item:prop:term2} are always satisfied when $L$ is finite.

Thm.~\ref{thm:pdrcor} and Prop.~\ref{prop:term} still hold if \textbf{Induction} rule is dropped. However, the rule can accelerate the convergence of KT sequences and improve efficiency.
\begin{remark}[LT-OpPDR]\label{rem:LTOpPDR}
 The GFP-UA problem $\alpha\le^{?}\nu F$ is the dual of LFP-OA, obtained by opposing the order $\le$ in $L$. We can also dualize the LT-PDR algorithm (Alg.~\ref{alg:pdr}), obtaining what we call the \emph{LT-OpPDR} algorithm for GFP-UA. Moreover, we can express LT-OpPDR as LT-PDR if a suitable \emph{involution} $\neg\colon L\to L^{\op}$ is present. See \finalarxiv{\cite[Appendix B]{ArxivFull}}{Appendix~\ref{appendix:LTOpPDR}} for further details; see also Prop.~\ref{prop:corresponds_inv}.
\end{remark}

\section{Structural Theory of PDR by Category Theory}\label{sec:strTh}
Before we discuss  concrete instances of LT-PDR in~\S{}\ref{sec:instances},
 we develop a structural theory of transition systems and predicate transformers as a basis of LT-PDR.
The theory is formulated in the language of \emph{category theory}~\cite{MacLane71,Awodey06,Jacobs16coalgBook,CLTT}.
We use category theory
 because 1) categorical modeling of relevant notions
is well established in the community (see e.g.~\cite{Jacobs16coalgBook,CLTT,BonchiKP18,AguirreK20,Sokolova11}), and 2)  it gives us the right level of abstraction that accommodates a variety of instances. In particular,  qualitative and quantitative settings are described in a uniform manner.

\begin{table}[tbp]
\caption{Categorical modeling of state-based dynamics and predicate transformers}
\label{table:categoricalNotions}
 \begin{tabular}{m{16em}l}
 \toprule
  \rowcolor[gray]{.87}
 \multicolumn{2}{c}{a transition system as a \emph{coalgebra}~\cite{Jacobs16coalgBook} in
the base category $\mathbb{B}$ of sets and functions}
 \\ \cmidrule{1-2}
 objects $X,Y,\dotsc$ in
  $\mathbb{B}$~~~
 &
 sets (in our examples where $\mathbb{B}=\mathbf{Set}$)
 \\
 \rowcolor[gray]{.95}
 an arrow $f\colon X\to Y$ in $\mathbb{B}$
 &
 a function  (in our examples where $\mathbb{B}=\mathbf{Set}$)
 \\
 a functor $G\colon\mathbb{B}\to \mathbb{B}$
 &
\parbox{20em}{ a transition type \\
$\left(\parbox{18em}{$G=\Pf$ for Kripke structures (\S\ref{sec:LTPDRsForKripke}), \\$G=(\D(-)+1)^\Act$ for MDPs (\S\ref{sec:LTPDRsForMDP}), etc.}\right)$
}
 \\
  \rowcolor[gray]{.95}
 a \emph{coalgebra} $\delta\colon S\to GS$ in $\mathbb{B}$~\cite{Jacobs16coalgBook}
 &
 a transition system (Kripke structure, MDP, etc.)
 \\
 \midrule
  \rowcolor[gray]{.87}
 \multicolumn{2}{c}{a \emph{fibration} $p\colon \mathbb{E}\to \mathbb{B}$~\cite{CLTT} that equips sets in $\mathbb{B}$ with  \emph{predicates}}
 \\ \cmidrule{1-2}
 the fiber category $\mathbb{E}_{S}$  over $S$ in $\mathbb{B}$
 &
 the lattice of predicates over a set $S$
 \\
  \rowcolor[gray]{.95}
\parbox{15em}{ the \emph{pullback} functor $l^{*}\colon \mathbb{E}_{Y}\to \mathbb{E}_{X}$
 \\\phantom{hoge} for $l\colon X\to Y$ in $\mathbb{B}$
}
 &
\parbox{20em}{ substitution $P(y)\mapsto P(l(x))$ in
  \\\phantom{hoge}predicates $P\in \mathbb{E}_{Y}$ over $Y$}
 \\[+.7em]
 a \emph{lifting} $\dot{G}\colon \mathbb{E}\to\mathbb{E}$ of $G$ along $p$
 &
\parbox{20em}{ logical interpretation of the transition type $G$
 \\
 (specifies e.g.\ the  may vs.\ must modalities)
}
 \\\midrule
  \rowcolor[gray]{.87}
 \multicolumn{2}{c}{the \emph{predicate transformer}, whose fixed points are of our interest}
 \\ \cmidrule{1-2}
 the composite $\delta^{*}\dot{G}\colon \mathbb{E}_{S}\to\mathbb{E}_{S}$
 &
\parbox{20em}{ the predicate transformer associated with
  \\\phantom{hoge} the transition system $\delta$
} \\\bottomrule
 \end{tabular}
\end{table}

 Our structural theory (\S\ref{sec:strTh}) serves as a backend, not a frontend. That is,
\begin{itemize}
 \item the theory in \S\ref{sec:strTh} is important in that it explains how the instances in  \S\ref{sec:instances} arise and why others do not, but
 \item the instances in \S\ref{sec:instances} are described in non-categorical terms, so readers who skipped \S\ref{sec:strTh} will have no difficulties following \S\ref{sec:instances}  and using those instances.
\end{itemize}

\subsection{Categorical Modeling of Dynamics and Predicate Transformers}\label{sec:categoricalModeling}
Our interests are in  instances of the LFP-OA problem $\mu F \leq^{?} \alpha$ (Def.~\ref{def:lfpOverapprox}) that appear in  \emph{model checking}. In this context, 1) the underlying lattice $L$ is that of \emph{predicates} over a state space, and 2) the function $F\colon L\to L$ arises from the dynamic/transition structure, specifically as a \emph{predicate transformer}. The categorical notions in Table~\ref{table:categoricalNotions} model these ideas (state-based dynamics, predicate transformers). This modeling is well-established in the community.

Our introduction of Table~\ref{table:categoricalNotions} here is minimal, due to the limited space.
See \finalarxiv{\cite[Appendix C]{ArxivFull}}{Appendix~\ref{appendix:categorical}} and the references therein for more details.

A \emph{category}  consists of \emph{objects} and \emph{arrows} between them. In Table~\ref{table:categoricalNotions}, categories occur twice: 1) a \emph{base category} $\mathbb{B}$ where objects are typically sets and arrows are typically functions; and 2) \emph{fiber categories} $\mathbb{E}_{S}$, defined for each object $S$ of $\mathbb{B}$, that are identified with the lattices of \emph{predicates}. Specifically, objects $P,Q,\dotsc$ of $\mathbb{E}_{S}$ are predicates over $S$, and an arrow $P\to Q$ represents logical implication. A general fact behind the last is that every preorder is a category---see e.g.~\cite{Awodey06}.

\myparagraph{Transition Systems as Coalgebras}
State-based transition systems are modeled as \emph{coalgebras} in the base category $\mathbb{B}$~\cite{Jacobs16coalgBook}. We use a \emph{functor} $G\colon\mathbb{B}\to\mathbb{B}$ to represent a transition type. A \emph{$G$-coalgebra} is an arrow $\delta\colon S\to GS$, where $S$ is a state space and $\delta$ describes the dynamics. For example,  a Kripke structure can be identified with a pair $(S,\delta)$ of a set $S$ and a function $\delta\colon S\to \Pf S$, where $\Pf S$ denotes the powerset. The powerset construction $\Pf$ is known to be a functor $\Pf\colon \Set\to\Set$; therefore Kripke structures are $\Pf$-coalgebras. For other choices of $G$, $G$-coalgebras become different types of transition systems, such as MDPs (\S\ref{sec:LTPDRsForMDP}) and Markov Reward Models (\S\ref{sec:LTPDRsForMRM}).

\myparagraph{Predicates Form a Fibration }
 Fibrations are powerful categorical constructs that can model various  indexed entities; see e.g.~\cite{CLTT} for its general theory. Our use of them is for organizing the lattices $\mathbb{E}_{S}$ of \emph{predicates} over a set $S$, indexed by the choice of $S$. For example, $\mathbb{E}_{S}=2^{S}$---the lattice of subsets of $S$---for modeling qualitative predicates. For quantitative reasoning (e.g.\ for MDPs), we use $\mathbb{E}_{S}=[0,1]^{S}$, where $[0,1]$ is the unit interval. This way,  qualitative and quantitative reasonings are mathematically unified in the language of fibrations.

A \emph{fibration} is a functor $p\colon \mathbb{E}\to\mathbb{B}$ with suitable properties;  it can be thought of as a collection $(\mathbb{E}_{S})_{S\in \mathbb{B}}$ of \emph{fiber categories} $\mathbb{E}_{S}$---indexed by objects $S$ of $\mathbb{B}$---suitably organized as a single category $\mathbb{E}$. Notable in this organization is that we obtain the \emph{pullback} functor $l^{*}\colon \mathbb{E}_{Y}\to \mathbb{E}_{X}$
for each arrow $l\colon X\to Y$ in $\mathbb{B}$. In our examples, $l^{*}$ is a \emph{substitution} along $l$ in predicates---$l^{*}$ is the monotone map that carries  a predicate $P(y)$ over $Y$ to the predicate $P(l(x))$ over $X$.

In this paper, we restrict to a subclass of fibrations (called \emph{\CLatw fibrations}) in which every fiber category $\mathbb{E}_{S}$ is a complete lattice, and each pullback functor preserves all meets. We therefore write $P\le Q$ for arrows in $\mathbb{E}_{S}$; this represents logical implication, as announced above. Notice that each $f^*$ has a left adjoint (lower adjoint in terms of Galois connection), %
which exists by Freyd's adjoint functor theorem. The left adjoint is denoted by $f_*$.

\begin{wrapfigure}[4]{r}{0pt}
\begin{math}
 \xymatrix@R=1em{
 \mathbb{E}
        \ar[r]^-{\dot{G}}
        \ar[d]_-{p}
 &
 \mathbb{E}
        \ar[d]_-{p}
 \\
 \mathbb{B}
        \ar[r]^-{G}
 &
 \mathbb{B}
 }
\end{math}
\end{wrapfigure}
We also consider a \emph{lifting} $\dot{G}\colon \mathbb{E}\to\mathbb{E}$ of $G$ along $p$; it is a functor $\dot{G}$ such that $p\dot{G}=Gp$. See the diagram on the right. It specifies the \emph{logical interpretation} of the transition type $G$. For example, for $G=\Pf$ (the powerset functor) from the above, two choices of $\dot{G}$ are for the \emph{may} and \emph{must} modalities. See e.g.~\cite{HermidaJ98,KoriHK21,AguirreK20,KomoridaKHKH19}.

\myparagraph{Categorical Predicate Transformer} The above constructs allow us to model predicate transformers---$F$ in our examples of the LFP-OA problem $\mu F \leq^{?} \alpha$---in categorical terms. A \emph{predicate transformer} along a coalgebra $\delta\colon S\to GS$ with respect to the lifting $\dot{G}$ is simply the composite
\begin{math}
 \mathbb{E}_{S}\xrightarrow{\dot{G}}
 \mathbb{E}_{GS}\xrightarrow{\delta^{*}}
 \mathbb{E}_{S}
\end{math}, where the first $\dot{G}$ is the restriction of $\dot{G}\colon \mathbb{E}\to\mathbb{E}$ to $\mathbb{E}_{S}$. Intuitively, 1) given a \emph{postcondition} $P$ in $\mathbb{E}_{S}$, 2) it is first interpreted as the predicate $\dot{G}P$ over $GS$, and then 3) it is pulled back along the dynamics $\delta$ to yield a \emph{precondition} $\delta^{*}\dot{G}P$.  Such (backward) predicate transformers are fundamental in a variety of model checking problems.

\subsection{Structural Theory of PDR from Transition Systems}\label{sec:structuralTheoryofPDR}
We formulate a few general \emph{safety} problems. We show how they are amenable to the LT-PDR (Def.~\ref{def:lt-pdr}) and LT-OpPDR (Rem.~\ref{rem:LTOpPDR})  algorithms.
\begin{definition}[\backwardGFP, BSP] \label{def:safe_prob_bd} Let $p$ be a \CLatw
  fibration, $\delta: S \to GS$ be a coalgebra in $\mathbb{B}$, and $\dot{G}: \mathbb{E} \to \mathbb{E}$ be a lifting of
  $G$ along $p$ such that $\dot G_X:\mathbb{E}_X\to\mathbb{E}_{GX}$ is
  $\omega^\op$-continuous for each $X\in\mathbb{B}$.  The
  \emph{\backwardGFP for
    $(\iota \in \mathbb{E}_S, \delta,
    \alpha \in \mathbb{E}_S)$ in $(p, G, \dot{G})$} is
    the GFP-UA problem for $(\mathbb{E}_S, \alpha \land \delta^*\dot{G}, \iota)$, that is,
  \begin{equation}
    \label{eq:bd_gfp}
    \iota \;\leq^?\; \nu x.\, \alpha \land \delta^*\dot{G} x.
  \end{equation}
\end{definition}
Here, $\iota$ represents the initial states and $\alpha$ represents the safe states.
The predicate transformer  $x\mapsto \alpha \land \delta^*\dot{G}x$ in (\ref{eq:bd_gfp}) is the standard one for modeling safety---currently safe ($\alpha$), and the next time $x$ ($\delta^*\dot{G}x$). Its gfp is the safety property; (\ref{eq:bd_gfp}) asks if all initial states ($\iota$) satisfy the safety property.
Since
the \backwardGFP is a GFP-UA problem, we can solve it by LT-OpPDR (Rem.~\ref{rem:LTOpPDR}).

\begin{wrapfigure}[5]{r}{0pt}
  \begin{math}
  \xymatrix@R=2em@C=3.3em{
  \text{BSP} \ar[r]^-{\text{as-is}} \ar@/_1ex/[dr]_-{\text{involution }\neg} \ar@/^1ex/[dr]^(.7){\text{ suitable adjoints}
} &\text{GFP-UA} \ar[r]^-{\text{LT-OpPDR}} &\text{True/False} \\
   &\text{LFP-OA} \ar[r]^-{\text{LT-PDR}} &\text{True/False}
  }
  \end{math}
\end{wrapfigure}
Additional assumptions allow us to
reduce the \backwardGFP to LFP-OA problems, which are solvable by LT-PDR, as shown on the right.

The first case requires the existence of the \emph{left adjoint} to the predicate transformer
$\delta^*\dot G_S:\mathbb{E}_S\to \mathbb{E}_S$.
Then we can translate BSP to the following LFP-OA problem.
It directly asks
 whether all reachable states are safe.
\begin{proposition}[\forwardLFP, FSP] \label{prop:safe_prob_fd} In the setting of
  Def.~\ref{def:safe_prob_bd}, assume that each
  $\dot{G}_X: \mathbb{E}_{X} \to \mathbb{E}_{GX}$ preserves all meets.
  Then by letting $\dot{H}_S: \mathbb{E}_{GS} \to \mathbb{E}_{S}$ be
  the left adjoint of $\dot{G}_S$, the BSP \eqref{eq:bd_gfp} is
  equivalent to the LFP-OA problem for
  $(\mathbb{E}_S,\iota\vee\dot{H}_S\delta_*,\alpha)$:
  \begin{equation} \label{eq:fd_lfp}
    \mu x.\,\iota \lor \dot{H}_S\delta_*x\; \leq^?\; \alpha.
  \end{equation}
  This problem
 is called the \emph{\forwardLFP} for $(\iota, \delta, \alpha)$ in $(p, G, \dot{G})$.
    \qed
\end{proposition}

The second case assumes that the complete
lattice $\mathbb{E}_S$ of predicates admits an involution operator
$\neg:\mathbb{E}_S\to\mathbb{E}_S^\op$ (cf.\ \finalarxiv{\cite[Appendix B]{ArxivFull}}{Appendix~\ref{appendix:LTOpPDR}}).
\begin{proposition}[\backwardLFP, IBSP] \label{prop:corresponds_inv} In the setting of
  Def.~\ref{def:safe_prob_bd}, assume further that there is a monotone
  function $\neg: \mathbb{E}_S \to \mathbb{E}^\op_S$ satisfying
  $\neg \circ \neg = \mathrm{id}$.  Then the \backwardGFP \eqref{eq:bd_gfp}
  is equivalent to the LFP-OA problem for
  $(\mathbb{E}_S,(\neg\alpha)\lor(\neg\circ\delta^*\dot
  G\circ\neg),\neg\iota)$, that is,
  \begin{equation}  \label{eq:bd_lfp}
    \mu x.\,(\neg \alpha) \lor (\neg\circ \delta^*\dot{G}\circ\neg x)\; \leq^{?}\; \neg \iota.
  \end{equation}
  We call
   \eqref{eq:bd_lfp} the \emph{\backwardLFP} for $(\iota, \delta, \alpha)$ in $(p, G, \dot{G})$.
  Here  $(\neg \alpha) \lor (\neg\circ \delta^*\dot{G}\circ\neg (-))$ is the
  \emph{inverse backward predicate transformer}.
  \qed
\end{proposition}

 When both additional assumptions are fulfilled (in Prop.~\ref{prop:safe_prob_fd} \&~\ref{prop:corresponds_inv}), we obtain two LT-PDR algorithms to
solve BSP.
One can even simultaneously run these two algorithms---this is done in
fbPDR~\cite{SeufertS18, SeufertS19}. See also \S\ref{sec:LTPDRsForKripke}.
\begin{auxproof}
 Such a setting does not seem common, however:    the only example we know is
 for Kripke structures~\cite{SeufertS18, SeufertS19}.
 We will illustrate this special case in
 \S\ref{sec:LTPDRsForKripke}.
\end{auxproof}

\section{Known and New PDR Algorithms as Instances}\label{sec:instances}
We present several concrete instances of our LT-PDR algorithms. The one for Markov reward models is new (\S\ref{sec:LTPDRsForMRM}). We also sketch how those instances can be systematically derived by the theory in \S\ref{sec:strTh}; details are in \finalarxiv{\cite[Appendix D]{ArxivFull}}{Appendix~\ref{appendix:strDerivDetails}}.

\subsection{LT-PDRs for Kripke Structures: \FKripke and \IBKripke}\label{sec:LTPDRsForKripke}
In most of the PDR literature, the target system is a Kripke structure that arises from a program's operational semantics. A \emph{Kripke structure} consists of a set $S$ of states
and a transition relation
$\delta\subseteq S\times S$ (here we ignore initial states and atomic propositions). The basic problem formulation is as follows.
\begin{definition}[backward safety problem (BSP) for Kripke structures]\label{def:BSPKripke}
 The \emph{BSP} for a Kripke structure $(S,\delta)$, a set $\iota\in 2^S$ of initial states, and a set $\alpha\in 2^S$ of safe states, is the  GFP-UA problem
\begin{math}\label{eq:BSPKripke}
   \iota \,\leq^?\, \nu x.\, \alpha \land F' x,
\end{math}
where $F'\colon 2^{S}\to 2^{S}$ is defined by $F'(A)\triangleq\{s\mid \forall s'.\, ((s,s')\in \delta \Rightarrow s'\in A)\}$.
\end{definition}
It is clear that the GFP in Def.~\ref{def:BSPKripke} represents the set of states from which all reachable states are in $\alpha$. Therefore the BSP is the usual safety problem.

The above BSP is easily seen to be equivalent to the following problems.
\begin{proposition}[forward safety problem (FSP) for Kripke structures]\label{prop:FSPKripke}
The BSP in Def.~\ref{def:BSPKripke} is equivalent to the  LFP-OA problem
\begin{math}\label{eq:FSPKripke}
  \mu x.\, \iota \lor \Fkrf x\, \leq^? \,\alpha
\end{math},
 where  $\Fkrf\colon 2^{S}\to 2^{S}$ is defined by $\Fkrf(A)\triangleq \bigcup_{s\in A}\{s'\mid (s,s')\in \delta\}$.
 \qed
\end{proposition}

\begin{proposition}[inverse backward safety problem (IBSP) for Kripke structures]\label{prop:IBSPKripke}
The BSP in Def.~\ref{def:BSPKripke} is equivalent to the  LFP-OA problem
\begin{math}\label{eq:IBSPKripke}
\mu x.\, \neg \alpha\lor \neg F'(\neg x) \,\leq^?\, \neg \iota
\end{math},
where $\neg\colon 2^{S}\to 2^{S}$ is the complement function $A\mapsto S\setminus A$.
 \qed
\end{proposition}

\myparagraph{Instances of LT-PDR}
The FSP and IBSP (Prop.~\ref{prop:FSPKripke}--\ref{prop:IBSPKripke}), being  LFP-OA, are amenable to the LT-PDR algorithm (Def.~\ref{def:lt-pdr}). Thus we obtain two instances of LT-PDR; we call them \emph{\FKripke} and \emph{\IBKripke}. \IBKripke{} is a step-by-step dual to the application of LT-OpPDR to the BSP (Def.~\ref{def:BSPKripke})---see Rem.~\ref{rem:LTOpPDR}.

We compare these two instances of LT-PDR
with   algorithms in the literature.
If we impose $|C_i|=1$
on each element $C_i$ of Kleene sequences, the \FKripke{} instance of LT-PDR coincides with
the conventional IC3/PDR~\cite{Bradley11, EenMB11}. In contrast,  \IBKripke{}  coincides with \emph{Reverse PDR} in~\cite{SeufertS18,SeufertS19}.
The parallel execution of \FKripke{} and \IBKripke{} roughly corresponds to
fbPDR~\cite{SeufertS18, SeufertS19}.
\begin{auxproof}
 All these algorithms---including LT-PDR---represent predicates (i.e.\ subsets) by logical formulas in their implementation, in order to exploit  efficient SAT solvers.
\end{auxproof}

\myparagraph{Structural Derivation} The equivalent problems (Prop.~\ref{prop:FSPKripke}--\ref{prop:IBSPKripke})
are derived systematically from the categorical theory in \S\ref{sec:structuralTheoryofPDR}. Indeed, using a lifting
 $\dot \Pf\colon 2^{S} \to 2^{\Pf S}$ such that $A\mapsto\{A'\mid A'\subseteq A\}$ (the \emph{must modality} $\Box$), $F'$ in Def.~\ref{def:BSPKripke} coincides with $\delta^{*}\dot\Pf$ in (\ref{eq:bd_gfp}). The above $\dot \Pf$ preserves meets (cf.\ the modal axiom $\Box(\varphi\land\psi)\cong\Box\varphi\land\Box\psi$, see e.g.~\cite{BlackburnRV01}); thus Prop.~\ref{prop:safe_prob_fd} derives the FSP. Finally, $\neg$ in Prop.~\ref{prop:IBSPKripke} allows the use of Prop.~\ref{prop:corresponds_inv}. More details are in
 \finalarxiv{\cite[Appendix D]{ArxivFull}}{Appendix~\ref{appendix:strDerivDetails}}.

\subsection{LT-PDR for MDPs: \IBMDP}\label{sec:LTPDRsForMDP}
The only known PDR-like algorithm for \emph{quantitative} verification is  \emph{PrIC3}~\cite{BatzJKKMS20} for Markov decision processes (MDPs). Here we instantiate LT-PDR for MDPs and compare it with PrIC3.

An \emph{MDP} consists of a set $S$
of states,
a set $\Act$ of actions
and a transition function $\delta$ mapping $s \in S$ and $a \in \Act$
to either $\ast$ (``the action $a$ is unavailable at $s$'') or a probability distribution $\delta(s)(a)$ over $S$.

\begin{definition}[IBSP for MDPs]\label{def:IBSPMDP}
 The \emph{inverse backward safety problem (IBSP)} for an MDP $(S,\delta)$, an initial state $s_{\iota}\in S$, a real number $\lambda\in[0,1]$, and a set $\alpha\subseteq S$ of safe states, is the  LFP-OA problem
\begin{math}\label{eq:IBSPMDP}
  \mu x.\, F'( x) \;\leq^?\;  d_{\iota,\lambda}
\end{math}.
Here $d_{\iota,\lambda}\colon S\to [0,1]$ is the predicate such  that $d_{\iota,\lambda}(s_{\iota})=\lambda$ and $d_{\iota,\lambda}(s)=1$ otherwise.  $F'\colon [0,1]^{S}\to [0,1]^{S}$ is defined by $  F' (d)(s) =1$ if $s\not\in\alpha$, and $  F' (d)(s) =     \max\{\sum_{s' \in S} d(s') \cdot \delta(s)(a)(s') \mid a \in
    \Act, \delta (s)(a) \neq \ast \}$ if $s\in\alpha$.
\end{definition}
The function $F'$ in Def.~\ref{def:IBSPMDP} is a \emph{Bellman operator} for MDPs---it takes the average of $d$ over $\delta(s)(a)$ and takes the maximum over $a$. Therefore the lfp in Def.~\ref{def:IBSPMDP} is the maximum reachability probability to $S\setminus \alpha$; the problem asks if it is  $\le \lambda$. In other words, it asks whether the \emph{safety} probability---of staying in $\alpha$ henceforth, under any choices of actions---is $\ge 1-\lambda$. This problem is the same as in~\cite{BatzJKKMS20}.

\myparagraph{Instance of PDR} The IBSP (Def.~\ref{def:IBSPMDP}) is LFP-OA and thus amenable to LT-PDR. We call this instance \emph{\IBMDP};
See \finalarxiv{\cite[Appendix E]{ArxivFull}}{Appendix~\ref{ap:mdp}} for details.

\IBMDP{}
shares many essences with PrIC3
\cite{BatzJKKMS20}. It uses the operator $F'$ in Def.~\ref{def:IBSPMDP}, which coincides with the one in \cite[Def.~2]{BatzJKKMS20}.
PrIC3 maintains \emph{frames}; they coincide with KT sequences in \IBMDP.

 Our Kleene sequences correspond to
 \emph{obligations} in PrIC3, modulo the following difference. Kleene sequences aim at a negative witness (\S\ref{sec:neg}), but they happen to help the positive proof efforts too (\S\ref{sec:int}); obligations in PrIC3 are solely for accelerating the positive proof efforts.
Thus, if PrIC3 cannot solve these efforts, we need to check whether obligations yield a negative witness.

\myparagraph{Structural Derivation} One can derive the IBSP (Def.~\ref{def:IBSPMDP}) from the categorical theory in \S\ref{sec:structuralTheoryofPDR}. Specifically, we first formulate the \emph{BSP}
\begin{math}
   \neg d_{\lambda} \;\leq^?\; \nu x.\, d_\alpha \land \delta^*\dot{\Gmdp }x
\end{math},
where $\dot{G}$ is a suitable lifting (of $G$ for MDPs, Table~\ref{table:categoricalNotions}) that combines average and minimum,
 $\neg \colon [0,1]^{S}\to[0,1]^{S}$ is defined by  $(\neg d)(s)\triangleq 1-d(s)$,
and $d_{\alpha}$ is such that $d_{\alpha}(s)=1$ if $s\in\alpha$ and $d_{\alpha}(s)=0$ otherwise. Using
 $\neg \colon [0,1]^{S}\to[0,1]^{S}$ in the above as an involution, we apply Prop.~\ref{prop:corresponds_inv} and obtain the IBSP (Def.~\ref{def:IBSPMDP}).

Another benefit of the categorical theory is that it  can tell us a forward instance of LT-PDR (much like \FKripke{} in \S\ref{sec:LTPDRsForKripke}) is unlikely for MDPs. Indeed, we showed in Prop.~\ref{prop:safe_prob_fd} that $\dot{G}'s$ preservation of meets is essential  (existence of a left adjoint is equivalent to meet preservation). We can easily show that our $\dot{G}$ for MDPs does not preserve meets. See \finalarxiv{\cite[Appendix G]{ArxivFull}}{Appendix~\ref{appendix:profwedge}}.

\subsection{LT-PDR for Markov Reward Models: \MRM}
\label{sec:LTPDRsForMRM}
We present a PDR-like algorithm for \emph{Markov reward models (MRMs)}, which seems to be new, as an instance of LT-PDR.
An  MRM consists of a set $S$ of states
and a transition function $\delta$
that maps $s \in S$ (the current state) and $c \in \mathbb{N}$ (the reward) to a function
$\delta(s)(c): S \to [0, 1]$;  the last represents the probability distribution
of next states.

We solve the following problem. We use $[0,\infty]$-valued predicates---representing accumulated rewards---where $[0,\infty]$ is the set of extended nonnegative reals.

  \begin{definition}[SP for MRMs]\label{def:SPMRM}
 The \emph{safety problem (SP)} for an MRM $(S,\delta)$, an initial state $s_{\iota}\in S$,  $\lambda\in[0,\infty]$, and a set $\alpha\subseteq S$ of safe states is
 \begin{math}
  \mu x.\, F'( x) \,\leq^?\,  d_{\iota,\lambda}
 \end{math}.
 Here $d_{\iota,\lambda}\colon S\to [0,\infty]$ maps $s_{\iota}$ to $\lambda$ and others to $\infty$, and $F'\colon [0,\infty]^{S}\to [0,\infty]^{S}$ is defined by
 \begin{math}
 F'(d)(s) = 0
 \end{math} if $s\not\in\alpha$, and
 \begin{math}
 F'(d)(s) =
    \sum_{s' \in S,c \in \mathbb{N}} (c+d(s')) \cdot
    \delta(s)(c)(s')
 \end{math} if $s\in\alpha$.
 \end{definition}

The function $F'$ accumulates expected reward in $\alpha$. Thus the problem asks if the expected accumulated reward, starting from $s_{\iota}$ and until leaving $\alpha$, is $\le \lambda$.

\myparagraph{Instance of PDR} The SP (Def.~\ref{def:SPMRM}) is  LFP-OA thus amenable to LT-PDR. We call this instance \emph{\MRM}. It seems new.
See \finalarxiv{\cite[Appendix F]{ArxivFull}}{Appendix~\ref{ap:mrm}} for details.

\myparagraph{Structural Derivation} The function $F'$ in Def.~\ref{def:SPMRM} can be expressed categorically as $F'(x)=d_\alpha \land \delta^*\dot{G}(x)$, where $d_{\alpha}\colon S\to [0,\infty]$ carries $s\in\alpha$ to $\infty$ and $s\not\in\alpha$ to $0$, and $\dot{G}$ is a suitable lifting that accumulates expected reward. However, the SP (Def.~\ref{def:SPMRM}) is \emph{not} an instance of the three general safety problems in \S\ref{sec:structuralTheoryofPDR}.
\begin{auxproof}
 ---the combination of $\mu$ and $\land$ in Def.~\ref{def:SPMRM} does not occur in \S\ref{sec:structuralTheoryofPDR}.
\end{auxproof}
Consequently, we expect that other instances of LT-PDR than \MRM{} (such as \FKripke{} and \IBKripke{} in \S\ref{sec:LTPDRsForKripke}) are hard for MRMs.

\section{Implementation and Evaluation}
\label{sec:implEval}

\myparagraph{Implementation \lstinline{LTPDR}}
We implemented LT-PDR  in Haskell. Exploiting Haskell's language features, it is succinct ($\sim$50 lines) and almost a literal translation of Alg.~\ref{alg:pdr} to Haskell. Its main part is presented in \finalarxiv{\cite[Appendix K]{ArxivFull}}{Appendix~\ref{appendix:code}}. In particular, using suitable type classes, the code is as abstract and generic as Alg.~\ref{alg:pdr}.

 Specifically, our implementation is a Haskell module named \lstinline{LTPDR}. It has two interfaces, namely the type class \lstinline{CLat $\tau$} (the lattice of predicates) and the type  \lstinline{Heuristics $\tau$} (the definitions of \textbf{Candidate}, \textbf{Decide}, and \textbf{Conflict}). The main function for LT-PDR is
 \lstinline{ltPDR :: CLat $\tau$ => Heuristics $\tau$ -> ($\tau$ -> $\tau$) -> $\tau$ -> IO (PDRAnswer $\tau$)}, where the second argument is for a monotone function $F$ of type \lstinline{$\tau$ -> $\tau$} and the last is for the safety predicate  $\alpha$.

 Obtaining concrete instances is easy by fixing $\tau$ and \lstinline{Heuristics $\tau$}.  A simple implementation of \FKripke{} takes  $15$ lines; a more serious SAT-based one for \FKripke{}  takes $\sim$130 lines; \IBMDP{} and \MRM{} take $\sim$80 lines each.

\myparagraph{Heuristics}
 We briefly discuss the heuristics, i.e.~how to choose $x \in L$ in \textbf{Candidate}, \textbf{Decide}, and \textbf{Conflict}, used in our experiments.
The heuristics of \FKripke{} is based on the conventional PDR~\cite{Bradley11}.
The heuristics of \IBMDP{} is based on
the idea of representing
the smallest possible $x$ greater than some real number $v \in [0, 1]$
(e.g.~$x$ taken in \textbf{Candidate})
as $x=v+\epsilon$, where
$\epsilon$ is a symbolic variable.
This implies that \textbf{Unfold} (or \textbf{Valid}, \textbf{Model}) is always applied in finite steps, which further guarantees finite-step termination for invalid cases and $\omega$-step termination for valid cases (see \finalarxiv{\cite[Appendix H]{ArxivFull}}{Appendix~\ref{ap:heuristics_mdp}} for more detail).
The heuristics of \MRM{} is similar to that of \IBMDP{}.

\myparagraph{Experiment Setting}
We experimentally assessed the performance of instances of \lstinline{LTPDR}.
The settings are as follows:  1.2GHz Quad-Core Intel Core i7 with 10 GB memory using Docker, for \IBMDP{};
 Apple M1 Chip with  16 GB memory
for the other.
The different setting is because we needed Docker to run PrIC3~\cite{BatzJKKMS20}.

\myparagraph{Experiments with \MRM{}}
Table~\ref{table:expMRM}
shows the results. We observe that \MRM{} answered correctly, and that the execution time is reasonable.  Further performance analysis (e.g.\ comparison with~\cite{KatoenKZ05}) and improvement is future work; the point here, nevertheless, is the fact that we obtained a reasonable MRM model checker by adding $\sim$80 lines to the generic solver \lstinline{LTPDR}.

\myparagraph{Experiments with \IBMDP{}}
Table~\ref{table:expMDP} shows the results.
Both PrIC3 and our \IBMDP{} solve a a linear programming (LP) problem in
 \textbf{Decide}.
PrIC3 uses Z3 for this; \IBMDP{} uses GLPK.
PrIC3  represents an MDP symbolically, while \IBMDP{} do so concretely. Symbolic representation in \IBMDP{} is possible---it is future work.
PrIC3 can use four different
 \emph{interpolation generalization} methods, leading to different performance (Table~\ref{table:expMDP}).

We observe that \IBMDP{} outperforms PrIC3 for some benchmarks with smaller state spaces.
We believe that the failure of \IBMDP{} in many  instances can be attributed to our current choice of a generalization method (it is the closest to the linear one for PrIC3).
Table~\ref{table:expMDP} suggests that  use of \emph{polynomial} or \emph{hybrid} can enhance the performance.

\myparagraph{Experiments with \FKripke{}}
Table~\ref{table:expFKripke} shows the results. The benchmarks are mostly
 from the HWMCC'15 competition~\cite{HWMCC15},
except for \texttt{latch0.smv}\footnote{\url{https://github.com/arminbiere/aiger}} and \texttt{counter.smv} (our own).

IC3ref vastly outperforms \FKripke{} in many instances.
This is hardly a surprise---IC3ref was developed towards superior performance, while \FKripke's emphasis is on its theoretical simplicity and genericity.
We nevertheless see that \FKripke{} solves some benchmarks of substantial size, such as \texttt{power2bit8.smv}.
This demonstrates the practical potential of LT-PDR, especially in view of the following improvement opportunities (we will pursue them as future work): 1) use of well-developed SAT solvers (we currently use \texttt{toysolver}\footnote{\url{https://github.com/msakai/toysolver}} for its good interface but we could use Z3); 2) allowing  $|C_i| > 1$, a technique discussed in \S\ref{sec:LTPDRsForKripke} and implemented in IC3ref but not in \FKripke{}; and 3) other small improvements,  e.g.\ in our CNF-based  handling of propositional formulas.

\myparagraph{Ablation Study}
To assess the value of the key concept of PDR (namely the \emph{positive-negative interplay} between the Knaster--Tarski and Kleene theorems (\S\ref{sec:int})), we compared \FKripke{} with the instances of positive and negative LT-PDR (\S\ref{sec:pos}--\ref{sec:neg}) for Kripke structures.

Table~\ref{table:expAblation} shows the results.
Note that
the value of the positive-negative interplay is already theoretically established; see  e.g.\ Prop.~\ref{prop:C_X} (the interplay detects executions that lead to nowhere).
This value was also experimentally witnessed: see \texttt{power2bit8.smv} and \texttt{simpleTrans.smv}, where the one-sided methods made wrong choices and timed out.
One-sided methods can be efficient when they get lucky (e.g.\ in \texttt{counter.smv}).
 LT-PDR may be  slower because of the overhead of running two sides, but that is a trade-off for the increased chance of termination.

\myparagraph{Discussion}
We observe that all of the studied instances exhibited at least reasonable performance. We note again that detailed
performance analysis and improvement is out of our current scope. Being able to derive these model checkers, with such a small effort as $\sim$100 lines of Haskell code each, demonstrates the value of our abstract theory and its generic Haskell implementation \lstinline{LTPDR}.

\begin{table}[tbp!]\footnotesize\centering
    \caption{experimental results for our \FKripke{}, \IBMDP{}, and \MRM{}}
    \begin{subtable}[t]{0.48\textwidth}
        \centering
     \caption{Results with \MRM{}. The MRM is from \cite[Example 10.72]{BaierK}, whose ground truth expected reward is $\frac{4}{3}$. The benchmarks ask if the expected reward (not known to the solver) is $\le 1.5$ or $\le 1.3$.
}\label{table:expMRM}
\scalebox{.7}{     \begin{tabular}{ccc}
    \toprule
    Benchmark  & Result  &  Time  \\
    \midrule
    $\textsc{DieByCoin}^{{\leq^?}1.5}$  & True  & \SI{6.01}{ms} \\
    $\textsc{DieByCoin}^{{\leq^?}1.3}$  & False  & \SI{43.1}{\micro s} \\
    \bottomrule
  \end{tabular}
}    \end{subtable}
    \hfill
    \begin{subtable}[t]{0.48\textwidth}
        \centering
  \caption{Results with \FKripke{} in comparison with IC3ref, a reference implementation of \cite{Bradley11} (\url{https://github.com/arbrad/IC3ref}). Both solvers answered correctly.
  Timeout (TO) is 600 sec.
 }\label{table:expFKripke}
 \scalebox{.7}{
  \begin{tabular}{ccccc}
    \toprule
    Benchmark  &$|S|$ & Result  & \FKripke & IC3ref  \\
    \midrule
    latch0.smv  &$2^3$ & True  & \SI{317}{\micro s} & \SI{270}{\micro s} \\
    counter.smv  &$2^5$ & False  & \SI{1.620}{s} & \SI{3.27}{ms} \\
    power2bit8.smv &$2^{15}$ & True  & \SI{1.516}{s} & \SI{4.13}{ms} \\
    ndista128.smv &$2^{17}$ & True  & TO & \SI{73.1}{ms} \\
    shift1add256.smv &$2^{21}$ & True  & TO & \SI{174}{ms} \\
    \bottomrule
  \end{tabular}
 }    \end{subtable}

    \vspace{.6\baselineskip}
    \begin{subtable}{\textwidth}
     \centering
   \caption{Results with \IBMDP (an excerpt of \finalarxiv{\cite[Table 3]{ArxivFull}}{Table~\ref{table:expMDPextended}}). Comparison is against PrIC3~\cite{BatzJKKMS20} with four
different interpolation generalization methods (none, linear, polynomial, hybrid). The benchmarks are from~\cite{BatzJKKMS20}.
$|S|$ is the number of states of the benchmark MDP. ``GT pr.'' is for the \emph{ground truth probability}, that is the reachability probability $\mathit{Pr}^{\mathit{max}}(s_\iota \models \diamond (S \setminus\alpha))$ computed outside the solvers under experiments. The solvers were asked whether the GT pr.\ (which they do not know) is $\le \lambda$ or not; they all answered correctly. The last five columns show the average execution time in seconds.  -- is for ``did not finish,'' for out of memory or timeout (600 sec.)
}\label{table:expMDP}
\scalebox{.7}{  \begin{tabular}{ccccccccc}
    \toprule
    Benchmark & $|S|$ &
    GT pr.\
    & $\lambda$ & \IBMDP{} &
    \multicolumn{4}{c}{PrIC3}
    \\\cmidrule{6-9}
    &&&&& none & lin. & pol. & hyb.
    \\
    \midrule
    \multirow{2}{*}{Grid} & \multirow{2}{*}{$10^2$} & \multirow{2}{*}{$1.2E^{-3}$} & 0.3 &0.31 & 1.31 & 19.34 & -- & --  \\
              &   &   & 0.2 &0.48 & 1.75 & 24.62 & -- & --\\
    \midrule
    \multirow{2}{*}{Grid} & \multirow{2}{*}{$10^3$} & \multirow{2}{*}{$4.4E^{-10}$} & 0.3 &122.29 & -- & -- & -- & -- \\
              &   &   & 0.2 &136.46 & -- & -- & -- & --\\
    \midrule
    \multirow{3}{*}{BRP} & \multirow{3}{*}{$10^3$} & \multirow{3}{*}{0.035} & 0.1 &-- & -- & -- & -- & -- \\
              &   &   & 0.01  &18.52 & 56.55 & 594.89 & -- & 722.38  \\
              &   &   & 0.005 &1.36  & 11.68 & 238.09 & -- & --   \\
      \midrule
    \multirow{4}{*}{ZeroConf} & \multirow{4}{*}{$10^4$} & \multirow{4}{*}{0.5} & 0.9 &-- & -- & -- & 0.58 & 0.51 \\
              &   &   & 0.75 & -- & -- & -- & 0.55 & 0.46 \\
              &   &   & 0.52 & -- & -- & -- & 0.48 & 0.46 \\
              &   &   & 0.45 &$<$0.1 & $<$0.1 & $<$0.1 & $<$0.1 & $<$0.1 \\
    \midrule
    \multirow{4}{*}{Chain} & \multirow{4}{*}{$10^3$} & \multirow{4}{*}{0.394} & 0.9 & -- & 72.37 & -- & 0.91 & 0.70 \\
              &   &   & 0.4 & -- & 80.83 & -- & 0.93 & -- \\
              &   &   & 0.35 & 177.12 & 115.98 & -- & -- & --  \\
              &   &   & 0.3 & 88.27 & 66.89 & 557.68 & -- & -- \\
    \midrule
    \multirow{4}{*}{DoubleChain} & \multirow{4}{*}{$10^3$} & \multirow{4}{*}{0.215} & 0.9 & -- & -- & -- & 1.83 & 1.99 \\
              &   &   & 0.3 & -- & -- & -- & 1.88 & 1.96 \\
              &   &   & 0.216 & -- & -- & -- & 139.76 & -- \\
              &   &   & 0.15 & 7.46 & -- & -- & -- & -- \\
    \bottomrule
  \end{tabular}}

    \end{subtable}

    \vspace{.6\baselineskip}
    \begin{subtable}{\textwidth}
     \centering
      \caption{Ablation experiments: LT-PDR (\FKripke{}) vs.~positive and negative LT-PDRs, implemented for the FSP for Kripke structures. The benchmarks are  as in Table~\ref{table:expFKripke}, except for a new micro benchmark \texttt{simpleTrans.smv}.
  Timeout (TO) is 600 sec.}
  \label{table:expAblation}
  \centering
 \scalebox{.7}{  \begin{tabular}{ccccc}
    \toprule
    Benchmark  & Result  &  LT-PDR & positive & negative \\
    \midrule
    latch0.smv & True & \SI{317}{\micro s} & \SI{1.68}{ms} & TO \\
    power2bit8.smv  & True & \SI{1.516}{s} & TO & TO \\
    counter.smv  & False & \SI{1.620}{s} & TO & \SI{2.88}{\micro s} \\
    simpleTrans.smv & False  & \SI{295}{\micro s} & TO & TO \\
    \bottomrule
  \end{tabular}
 }    \end{subtable}

    \vspace{.6\baselineskip}
\end{table}

 \section{Conclusions and Future Work}
 We have presented a lattice-theoretic generalization of the PDR
 algorithm called LT-PDR. This involves the decomposition of the PDR
 algorithm into positive and negative ones, which are tightly connected
 to the Knaster--Tarski and Kleene fixed point theorems, respectively. We then combined it
 with the coalgebraic and fibrational theory for modeling transition
 systems with predicates. We instantiated it with several transition
 systems, deriving existing PDR algorithms as well as a new one over
 Markov reward models.
 We leave instantiating our LT-PDR and categorical safety problems to
 derive other PDR-like algorithms, such as PDR for hybrid
 systems~\cite{SuenagaI20}, for future work.

We will also work on the combination of our work and the theory of \emph{abstract interpretation}~\cite{CousotC77, CousotC79}.
Our current framework axiomatizes what is needed of heuristics, but it does not tell how to realize such heuristics (that differ a lot in different concrete  settings).
We expect abstract interpretation to  provide some general recipes for realizing such heuristics.

\clearpage

\bibliographystyle{splncs04}
\bibliography{mybib}

\ifarxiv
\clearpage
\appendix

\section{Further Discussion of Related Work}\label{appendix:relatedWorkOnGen}
We  discuss some other works on generalization of  PDR. Hoder and Bj{\o}rner~\cite{HoderB12} gave an abstract formulation of (the original) PDR, abstracting away implementation details (such as SAT-related ones) and presenting the algorithm itself as a transition system (an ``abstract transition system'' as they call it).
Their notion of
predicate transformer
 is an instance of our
forward predicate transformer (Prop.~\ref{prop:safe_prob_fd}).
They also identified an invariant of frames, and our definition of KT
sequence (Def.~\ref{def:kt_sequence}) is inspired by it.
Another theoretical study of PDR is by Rinetzky and
Shoham~\cite{RinetzkyS16}.  They studied  PDR using
abstract interpretation and showed a mapping between PDR
configurations and elements of what they call cartesian trace semantics.
In both of these works~\cite{HoderB12,RinetzkyS16}, the formulated PDR algorithms target at
Kripke structures, and do not accommodate quantitative verification.
They are instances of our LT-PDR especially for categorical safety problems introduced in \S{}\ref{sec:structuralTheoryofPDR} (specifically the FSP in~\S{}\ref{sec:LTPDRsForKripke}), similarly to the original PDR. Moreover, our view of PDR as collaborative searches for KT and Kleene witnesses is not explicit in~\cite{HoderB12,RinetzkyS16}.

\section{LT-OpPDR (Rem.~\ref{rem:LTOpPDR}) \conf{75}}
\label{appendix:LTOpPDR}

Recall that the GFP-UA problem $\alpha\le^{?}\nu F$ for $(L,F,\alpha)$ is defined to be the
LFP-OA problem for $(L^\op,F,\alpha)$. Hence we can solve the
GFP-UA problem by executing the LT-PDR algorithm over $L^\op$. We call
this algorithm \emph{LT-OpPDR}; in other words,
LT-OpPDR is obtained by opposing each inequality in LT-PDR.

Although LT-OpPDR is a formal dual of LT-PDR,
applying the PDR-like algorithm for solving GFP-UA problems seems to
be new.

 When $L$ admits a duality by involution $\neg:L\to L^\op$, the GFP-UA
problem in $L$ can be formulated as the LFP-OA problem in $L$ (not in $L^\op$ as in the above).
\begin{proposition} \label{prop:corresponds_pdr_dpdr}
  Let $L,F,\alpha$ be as in  Def.~\ref{def:lfpOverapprox}.
  Assume that
  there is a monotone function
  $\neg: L \to L^\op$
  satisfying $\neg\circ \neg=\mathrm{id}_L$.  Then the GFP-UA
  problem $\alpha\leq^{?}\nu F$ in $L$ is equivalent to the LFP-OA problem
  $\mu (\neg F \neg (-)) \leq^{?} \neg \alpha$ in $L$.
\end{proposition}
\begin{proof}
  This is a consequence of a more general statement about translating
  LFP-OA problem by isomorphisms. Let $L$ be a complete lattice,
  $\alpha$ be an element in $L$, and $F: L \to L$ be an
  $\omega$-continuous function. For any complete lattice $L'$ with an
  order-preserving isomorphism $f: L \iso L'$, LFP-OA problem
  $\mu F \leq^? \alpha$ in $L$ is equivalent to LFP-OA problem
  $\mu (f \circ F \circ f^{-1}) \leq^? f(\alpha)$ in $L'$. \qed
\end{proof}
In this case, we can invoke the LT-PDR algorithm over
$(L,{\neg}\circ F\circ {\neg},\neg\alpha)$ to solve the GFP-UA problem $\alpha \leq^? \nu F$.
We however note that the execution steps of LT-OpPDR over $(L, F, \alpha)$, i.e.~LT-PDR over
$(L^\op,F,\alpha)$, and the execution steps of LT-PDR over
$(L,{\neg}\circ F\circ{\neg},\neg\alpha)$ are essentially the same; the
configuration at each execution step is mutually convertible by the
involution $\neg$.

\section{Structural Theory of PDR by Category Theory, Further Categorical Preliminaries }
\label{appendix:categorical}
Here we provide more details on the categorical modeling in~\S\ref{sec:categoricalModeling}.

A fibration $p: \mathbb{E} \to \mathbb{B}$ is a functor that models
indexing and substitution.  That is, a functor
$p: \mathbb{E} \to \mathbb{B}$ can be seen as a family of categories
$(\mathbb{E}_X)_{X \in \mathbb{B}}$ indexed by $\mathbb{B}$-objects.
Categories with different indices are connected by \emph{substitution
  functors}.  In our examples, the base category $\mathbb{B}$ is that
of sets and functions; and the total category $\mathbb{E}$ models
``predicates'' over $\mathbb{B}$ objects.  We review a minimal set of
definitions and results on fibrations. A good reference is
\cite{CLTT}, here we quote some definitions and examples
given in \S{}2.1 of  \cite{KoriHK21}; see also \cite{AguirreK20} and
\cite{SprungerKDH18}.

\begin{wrapfigure}[6]{r}{0.42\textwidth}
  \begin{math}
    \xymatrix@R=.1em@C-1.7em{
        \mathbb{E} \ar[ddd]_-{p}
        & R
        \ar@/^/[rrd]^-{g}
        \ar@{.>}[rd]_-{h}
        \\
        &
        &
        P \ar[r]_-{f}
        &
        Q
        &
        l^{*}Q
        \ar[r]^-{\overline{l}}
        &
        Q
        \\
        & pR
        \ar@/^/[rrd]^-{pg}
        \ar@{->}[rd]_-{k}
        \\
        \mathbb{B}
        &
        &
        pP \ar[r]_-{pf}
        &
        pQ
        &
        X
        \ar[r]_-{l}
        &
        pQ
      }
  \end{math}
\end{wrapfigure}
\mbox{}\vspace{-\baselineskip}
\begin{definition}
  \!\!{\bf(fibre, fibration;} {\bf \cite[\S{}2.1]{KoriHK21})}
  Let $p: \mathbb{E} \to \mathbb{B}$ be a functor.
  For each $X\in\mathbb{B}$, the \emph{fibre} $\mathbb{E}_X$ over $X$ is the category with objects
  $P \in \mathbb{E}$ such that $pP = X$ and morphisms $f: P \to Q$
  such that $pf=\mathrm{id}_{X}$.

  A morphism $f: P \to Q$ in $\mathbb{E}$ is \emph{cartesian} if it
  satisfies the following universality: for each $g: R \to Q$ in
  $\mathbb{E}$ and $k: pR \to pP$ in $\mathbb{B}$ with
  $pg = pf \circ k$, there exists a unique morphism $h: R \to P$
  satisfying $g = f \circ h$ and $ph = k$ (see the diagram above).

  The functor $p: \mathbb{E} \to \mathbb{B}$ is a \emph{fibration} if,
  for each $Q \in \mathbb{E}$ and each $l: X \to pQ$ in $\mathbb{B}$,
  there exists $l^*Q\in \mathbb{E}$ and a morphism
  $\overline{l}: l^*Q \to Q$ such that $p\overline{l}=l$ and
  $\overline{l}$ is cartesian.
  The functor $p: \mathbb{E} \to \mathbb{B}$ is an \emph{opfibration}
  if $p^\op: \mathbb{E}^\op \to \mathbb{B}^\op$ is a fibration.  A
  functor that is both a fibration and an opfibration is called a
  \emph{bifibration}.
\end{definition}

\noindent
When $p$ is a fibration, the correspondence from $Q$ to $l^* Q$
described above induces the \emph{substitution functor}
$l^*: \mathbb{E}_Y \to \mathbb{E}_X$ which replaces the index.  The
following characterization of bifibrations is useful for us: a
fibration $p$ is a bifibration if and only if each substitution
functor $l^*: \mathbb{E}_Y \to \mathbb{E}_X$ (often called a
\emph{pullback}) has a left adjoint
$l_*: \mathbb{E}_X \to \mathbb{E}_Y$ (often called a
\emph{pushforward}).

\begin{definition}[lifting {\cite[\S{}2.1]{KoriHK21}}]\label{def:lifting}
  Let $p: \mathbb{E} \to \mathbb{B}$ be a functor.
  We say that an endofunctor $\dot{G}$
  on $\mathbb{E}$ is a \emph{lifting} of $G$ along $p$ if
  $p \circ \dot{G} = G \circ p$.  For an object
  $S\in\mathbb{B}$,
  we write $\dot G_S:\mathbb{E}_S\to\mathbb{E}_{GS}$ for
  the restriction of $\dot G$ to fibres.
\end{definition}

To manipulate complete lattices along a transition function, we focus
on a certain class of posetal fibrations called \CLatw fibrations.
They can be seen as \emph{topological functors}~\cite{Herrlich74}
whose fibres are posets.  Many categories arising from spacial and
logical structures naturally determine \CLatw fibrations.

\begin{definition}[\CLatw fibration {\cite[\S{}2.1]{KoriHK21}}] \label{def:clat} A \emph{\CLatw
    fibration} is a fibration $p\colon\mathbb{E}\to\mathbb{B}$ such that
  each fibre $\mathbb{E}_{X}$ is a complete lattice and each
  substitution $f^{*}\colon \mathbb{E}_{Y}\to \mathbb{E}_{X}$
  preserves all meets $\bigwedge$.
  In each fibre $\mathbb{E}_X$, the order is denoted by $\leq_X$ or
  $\leq$.  Its least and greatest elements are denoted by $\bot_X$ and
  $\top_X$; its join and meet are denoted by $\bigvee$ and
  $\bigwedge$.
\end{definition}
The above simple axioms of \CLatw fibrations induce many useful
structures~\cite{KomoridaKHKH19, SprungerKDH18}.  One of them is that
a \CLatw fibration is always a bifibration whose pushforwards $f_{*}$
arise essentially by Freyd's adjoint functor theorem.

\begin{example}[\CLatw fibration {\cite[\S{}2.1]{KoriHK21}}]\label{eg:clat}
 ($\mathbf{Pred} \to \Set$) %
The forgetful functor $\mathbf{Pred} \to \Set$ is a \CLatw fibration.
    Here $\mathbf{Pred}$ is the category  of predicates:
    objects are sets with predicates $(P \subseteq X)$, and morphisms
    $f: (P \subseteq X) \to (Q \subseteq Y)$ are functions
    $f: X \to Y$ satisfying $f(P) \subseteq Q$.

   \noindent
    \begin{minipage}[t]{.78\textwidth}
     \quad (Domain fibration $d^{\Omega}$) For each
      complete lattice $\Omega$, we introduce a \CLatw fibration
      $d^{\Omega}: \Set/\Omega \to \Set$ as follows.

      Here, we write $\Set/\Omega$ for the lax slice category
      with objects
    \end{minipage}
    \hfill
    \begin{minipage}[t]{.18\textwidth}
      \begin{math}
        \xymatrix@R=1em@C=0.5em{ {X} \ar[rr]^-{h}_-{\leq_X}
          \ar[rd]_-{f} && {Y} \ar[ld]^-{g}
          \\
          &\Omega&}
      \end{math}
    \end{minipage}
    $(X,f: X \to \Omega)$ of a set
    and a function (an ``$\Omega$-valued predicate on $X$'').
    We shall
    often write simply $f: X \to \Omega$ for the pair $(X,f)$.  Its
    morphisms from $f: X\to \Omega$ to $g: Y \to \Omega$ are functions
    $h: X \to Y$ such that $f \leq_X g \circ h$, as shown above,
    where the order $\leq_X$ is the pointwise order between functions
    of the type $X \to \Omega$; the same order $\leq_{X}$ defines the
    order in each fiber $(\Set/\Omega)_{X}=\Set(X,\Omega)$.
    Then $d_\Omega$ is the evident forgetful functor, extracting the
    upper part of the above triangle.  Following
    \cite[Def.~4.1]{AguirreK20}, we call $d^\Omega$ a \emph{domain
      fibration} (from the lax slice category).

\begin{auxproof}
     Notational convention: for a lattice element $x\in\Omega$,
    we write $\mathbf{x}$ for the constant $\Omega$-valued predicate
    ${\bf x}\triangleq \lambda{y}.x$.
\end{auxproof}

\end{example}

\section{Structural Derivation of Instances of LT-PDR (\S\ref{sec:instances})}
\label{appendix:strDerivDetails}
In \S\ref{sec:instances}, for each instance of LT-PDR, we only sketched its structural derivation from the categorical theory in \S\ref{sec:structuralTheoryofPDR}. Here we give a systematic exposition to the structural derivation.

\begin{wrapfigure}[5]{r}{0pt}
  \begin{math}
    \vcenter{\xymatrix@R=1.8em@C=1em{
        \Set/\Omega \lloop{\dot{G}} \ar[d]^{d^\Omega} \\
        \Set \lloop{G} }}
  \end{math}
\end{wrapfigure}
We discuss concrete instances of our  PDR framework. In its course, known PDR variations are organized in a unified categorical language; we also derive a new variation.

These concrete instances are formulated in a domain fibration $d^\Omega$ for varying $\Omega$ (Ex.~\ref{eg:clat}; see right).
Given a complete lattice $\Omega$, a set functor $G$, and a monotone $G$-algebra $\tau: G\Omega \to \Omega$ (see Def.~\ref{def:mono} below), we obtain a setting $(d^\Omega, G, \dot{G})$ for safety problems (\S\ref{sec:structuralTheoryofPDR}).
Specifically,
$\dot{G}$ is the lifting of $G$ defined by the given monotone $G$-algebra $\tau$, see Lem.~\ref{lem:mono_alg} below.

\begin{definition}[monotone algebra
  \cite{AguirreK20}] \label{def:mono} Let $G: \Set \to \Set$ be a
  functor and $\Omega$ be a complete lattice.  We call
  $\tau: G\Omega \to \Omega$ a \emph{monotone $G$-algebra} over
  $\Omega$ if
  $i \leq_X i' \Rightarrow \tau \circ Gi \leq_{GX} \tau \circ Gi'$
  holds for all $X\in\Set$ and all $i, i' \in \Set(X, \Omega)$.
\end{definition}

\begin{lemma}[\cite{BonchiKP18}] \label{lem:mono_alg}
  In the setting of Def.~\ref{def:mono},
  a monotone $G$-algebra $\tau$ gives rise to the lifting
  $\dot{G}: \Set/\Omega \to \Set/\Omega$ given by
  $\dot G(x)=\tau\circ Gx$.
  \qed

\end{lemma}

One benefit of this framework $(d^\Omega, G, \dot{G})$ is that we may
easily get an involution appeared in Prop.~\ref{prop:corresponds_inv}.  From a monotone function
$\neg: \Omega \to \Omega^\op$ satisfying
$\neg \circ \neg = \mathrm{id}$, we can define
$\neg: \mathbb{E}_S \to \mathbb{E}_S^\op$ mapping $f: S \to \Omega$ to
$\neg \circ f: S \to \Omega$.
All involutions appeared in this section can be defined in this way.

\subsection{LT-PDR for Kripke Structures:  \FKripke{} and \IBKripke{}
  \conf{70}
} \label{appendix:p_monad} We
instantiate the structural theory in
\S{}\ref{sec:structuralTheoryofPDR} to derive LT-PDR algorithms for
Kripke structures.
We then compare them with IC3/PDR~\cite{Bradley11, EenMB11} and Reverse
PDR~\cite{SeufertS18,SeufertS19}.

In most of the PDR literature, the target system is a Kripke structure that arises from a program's operational semantics.
A \emph{Kripke structure} consists of a set $S$ of states
and a transition relation
$\delta\subseteq S\times S$ (we ignore initial states and atomic propositions).
We regard $\delta$ as a function
$S \to \Pf S$; it is thus a coalgebra of the powerset functor $\Pf$ on $\Set$.

\begin{wrapfigure}[4]{r}{0pt}
  \begin{math}
    \vcenter{\xymatrix@R=1.8em@C=1em{
        \Pred \lloop{\dotP} \ar[d]^{d^\mathbf{2}} \\
        \Set \lloop{\Pf} }}
  \end{math}
\end{wrapfigure}
To employ the theory in  \S{}\ref{sec:structuralTheoryofPDR}, we
next choose a complete lattice and a monotone $\Pf$-algebra.
Consider the complete lattice $\mathbf{2} = \{\bot < \top\}$, and the monotone
$\Pf$-algebra $\tau_{\mathrm{and}}: \Pf \mathbf{2} \to \mathbf{2}$ mapping $A$ to
$\bigwedge_{a \in A} a$.
Then we obtain the triple $(d^\mathbf{2}, \Pf, \dotP)$ as a setting of  the safety problems in \S\ref{sec:structuralTheoryofPDR}.  Note that $\Set/\mathbf{2}$ is
isomorphic to the category of predicates $\mathbf{Pred}$ (see
Ex.~\ref{eg:clat}).

We are ready to consider the \backwardGFP for the transition system.
Let $\alpha \subseteq S$ be a set of safe states.
The \backwardGFP for $(\iota, \delta, \alpha)$ in
$(d^\mathbf{2}, \Pf, \dotP)$ is the GFP-UA problem
\begin{equation}
  \label{eq:bd_p}
  \iota \;\leq^?\; \nu x.\, \alpha \land \delta^*\dotP x.
\end{equation}
This is the problem checking whether the initial states are always in the safe states after any steps.
This setting satisfies both of the additional assumptions imposed in Prop.~\ref{prop:safe_prob_fd} and \ref{prop:corresponds_inv} (namely the existence of an adjoint and an involution).
Therefore, we have two LT-PDR algorithms to solve \eqref{eq:bd_p} as below.
\myparagraph{FSP for Kripke Structures (Prop.~\ref{prop:FSPKripke})}
Notice that the lifting $\dotP_S$, which maps $X \subseteq S$ to
$\Pf X \subseteq \Pf S$, has a left adjoint
$\bigcup_S: \Pred_{\Pf S} \to \Pred_S$.
Then by Prop.~\ref{prop:safe_prob_fd}, the \backwardGFP
\eqref{eq:bd_p} can be solved by the LT-PDR algorithm for the \forwardLFP
\begin{equation}
  \label{eq:fd_p}
  \mu x.\, \iota \lor \bigcup\delta_*x\; \leq^? \;\alpha.
\end{equation}
  The forward predicate transformer
  $\Fkrf = \iota \lor \bigcup\delta_*(-)$ on predicates over $S$ expands
  to $\Fkrf(x) = \iota \cup \bigcup_{s \in x} \delta(s)$,
  which has already appeared in Example~\ref{ex:forward}.

\myparagraph{IBSP for Kripke Structures (Prop.~\ref{prop:IBSPKripke})}
Notice that the complete lattice of predicates over $S$
  admits an involution defined by set complement
  $\neg x \triangleq S \setminus x$.
  Then by
  Prop.~\ref{prop:corresponds_inv},
  the \backwardGFP \eqref{eq:bd_p}
  can be solved by the LT-PDR algorithm for the \backwardLFP
  \begin{equation} \label{eq:rev_p} \mu x.\, \neg \alpha \lor
    \neg \delta^*\dotP(\neg x) \;\leq^?\; \neg \iota.
  \end{equation}
  The function $\Fkrb = \neg \alpha \lor \neg \delta^*\dotP\neg(-)$ that appears in the above expands to
  $
    \Fkrb(x) = (S \setminus \alpha) \cup \{s \mid \exists s' \in \delta s.~s' \in x \}.
  $

\subsection{LT-PDR for MDPs: \IBMDP}
\label{appendix:mdp}
We instantiate the theory  in
\S{}\ref{sec:structuralTheoryofPDR} to derive an LT-PDR algorithm for Markov
decision processes (MDP for short).  We then compare it with the
probabilistic model checking algorithm \emph{PrIC3}
\cite{BatzJKKMS20}.

An MDP consists of a set $S$
of states,
a set $\Act$ of actions
and a transition function $\delta$ mapping $s \in S$ and $a \in \Act$
to $\delta(s)(a)$ representing a
probability distribution of next states.
We model the transition function of the MDP as a coalgebra $\delta: S \to \Gmdp S$ of
$\Gmdp \triangleq (\D(-)+1)^\Act$, where $\D$ is the
finite probability distribution endofunctor on $\Set$
\cite{Sokolova11}.
The case $\delta(s)(a) = \ast \in 1$ means that the action $a$ is not available at $s$.

\begin{wrapfigure}[5]{r}{0pt}
  \begin{math}
    \vcenter{\xymatrix@R=1.8em@C=1em{
        \Set/[0, 1] \lloop{\dot{\Gmdp}} \ar[d]^{d^{[0, 1]}} \\
        \Set \lloop{\Gmdp} }}
  \end{math}
\end{wrapfigure}
To employ the theory in \S{}\ref{sec:structuralTheoryofPDR}, we
next choose a complete lattice $\Omega$ and a monotone $\Gmdp
$-algebra over $\Omega$. Consider the complete lattice $[0, 1]$ of the real numbers in the unit
interval with the usual order, and the monotone algebra
$\tau: \Gmdp [0,1] \to [0, 1]$ mapping $f$ to
$\min \{\sum_{n \in [0, 1]} n \cdot f(a)(n) \mid a \in \Act, fa \neq
\ast\}$ (note that $\min \{\} = 1$).
Then we obtain the triple $(d^{[0, 1]}, \Gmdp ,\dot\Gmdp )$ as a setting of  the safety problems in \S\ref{sec:structuralTheoryofPDR}.
We note that $\dot{G}_S$ does not have a left adjoint (see Appendix~\ref{appendix:profwedge}). We therefore cannot apply Prop.~\ref{prop:safe_prob_fd} to the current setting.

We are ready to consider the \backwardGFP for MDPs.  Let
$s_\iota\in S$ be an initial state, and
$\alpha \subseteq S$ be a set of safe states. We convert
 $s_\iota$, $\lambda$ and $\alpha$ to $[0,1]$-valued predicates
$d_{\iota,\lambda}$ and $d_\alpha$: $d_{\iota,\lambda}$ maps
$s_\iota$ to $\lambda$ and others to $1$, and $d_\alpha$ maps
$s \in \alpha$ to $1$ and $s \not \in \alpha$ to $0$.  We use the involution
 $\neg \colon [0,1]^{S}\to[0,1]^{S}$  defined by  $(\neg d)(s)\triangleq 1-d(s)$, too. Then
the backward
safety problem for $(\neg d_{\iota,\lambda}, \delta, d_\alpha)$ in
$(d^{[0, 1]},\Gmdp , \dot{\Gmdp })$ is the GFP-UA problem
\begin{equation}
  \label{eq:mdpgfpua}
\neg  d_{\iota,\lambda} \;\leq^?\; \nu x.\, d_\alpha \land \delta^*\dot{\Gmdp }x.
\end{equation}
This is the problem whether the probability of being at $\alpha$ all
the time is greater than or equal to $1-\lambda$ under any choices of actions in the MDP.

\myparagraph{IBSP for MDPs (Def.~\ref{def:IBSPMDP})}
Note that the complete lattice of $[0,1]$-valued predicates over $S$
admits the above involution $\neg$.
Then by
Prop.~\ref{prop:corresponds_inv},
the BSP in
  \eqref{eq:mdpgfpua} can be solved by  LT-PDR  for the inverse backward
  safety problem
\begin{equation}
  \label{eq:mdplfpoa}
  \mu x.\, \neg d_\alpha \lor \neg \delta^*\dot{\Gmdp }(\neg x) \;\leq^?\;  d_{\iota,\lambda}.
\end{equation}
The precise algorithm is in Appendix~\ref{ap:mdp}.
The function $\Fmdp \triangleq \neg d_\alpha \lor \neg \delta^*\dot{\Gmdp }(\neg
-)$
used in~(\ref{eq:mdplfpoa}) expands as follows (recall
$\delta^*\dot{\Gmdp }(d) = \tau \circ \D d \circ \delta$):
\begin{equation}
  \Fmdp (d)(s) =
  \begin{cases}
    1 & (s \not \in \alpha) \\
    \max\{\sum_{s' \in S} ds' \cdot \delta(s)(a)(s') \mid a \in
    \Act, \delta (s)(a) \neq \ast \} & (s \in \alpha).
  \end{cases}
  \label{eq:rbpt}
\end{equation}
This is a standard Bellman Operator for MDPs.

\subsection{LT-PDR for Markov Reward Models: \MRM}
\label{appendix:mrm}
We instantiate an LT-PDR algorithm for Markov Reward Models (MRM for short),
which is seemingly new.
As we said in \S{}\ref{sec:LTPDRsForMRM},
the safety problem we will define is not an instance of the theory in
\S{}\ref{sec:structuralTheoryofPDR} (especially Prop.~\ref{def:safe_prob_bd}--\ref{prop:corresponds_inv}).

An MRM consists of a set $S$ of states
and a transition function $\delta$
mapping $s \in S$ and $c \in \mathbb{N}$ to a function
$\delta(s)(c): S \to [0, 1]$ that represents probability distribution
of next states.  We model the transition function of the MRM as a
coalgebra $\delta: S \to GS$ of the endofunctor
$G \triangleq \mathcal{D}((-) \times \mathbb{N})$, where $\mathcal{D}$
is introduced in \S{}\ref{appendix:mdp}.

Note that the above definition accommodates another definition of MRM, namely as an MC
$(S, \delta_{\mathrm{MC}}: S \to \D S)$ with a reward
function $\mathrm{rew}: S \to \mathbb{N}$~\cite{BaierK}. Specifically, we can set
$\delta(s)(c)(s')\coloneqq\delta_\mathrm{MC}(s)(s')$ if $c=\mathrm{rew}(s)$
and $\delta(s)(c)(s')\coloneqq 0$ otherwise.

\begin{wrapfigure}[5]{r}{0pt}
  \begin{math}
    \vcenter{\xymatrix@R=1.8em@C=1em{
        \Set/[0, \infty] \lloop{\dot{G}} \ar[d]^{d^{[0, \infty]}} \\
        \Set \lloop{G} }}
  \end{math}
\end{wrapfigure}
We next define a complete lattice and a monotone $G$-algebra.  Consider
the complete lattice $[0, \infty]$ of the extended nonnegative real
numbers with the usual order, and the monotone $G$-algebra
$\tau: G[0, \infty] \to [0, \infty]$ defined by
$ \tau(\mu)=E_{(r,n)\sim \mu}[r+n]$. It takes the expectation of $r+n$
for the distribution
$\mu\in G[0,\infty]=\mathcal{D}([0,\infty]\times\mathbb{N})$
(see \cite[Ex.~6.6]{AguirreK20} for further details).  Then we obtain
the triple $(d^{[0, \infty]}, G, \dot{G})$, as a setting of the safety problem.

\myparagraph{SP for MRMs (Def.~\ref{def:SPMRM})}
Using the above data, we obtain the function $F'(x)=d_\alpha \land \delta^*\dot{G}(x)$ in \S\ref{sec:LTPDRsForMRM}. This $F'$ can be concretely described as in Def.~\ref{def:SPMRM}.

\section{LT-PDR Algorithm \IBMDP{}  for MDPs in \S{}\ref{sec:LTPDRsForMDP}} \label{ap:mdp}

See Algorithm \ref{alg:mdp}.
It solves the IBSP for MDPs (Def.~\ref{def:IBSPMDP}).

\begin{algorithm}
  \SetKwInOut{Input}{Input} \SetKwInOut{Output}{Output}
  \SetKwInOut{Initially}{Initially}
  \Input{$(d_{\iota,\lambda}, \delta, d_\alpha)$}
  \Output{`True' with a conclusive KT sequence, or `False' with a conclusive Kleene sequence}
  \KwData{$(X; C)$ where $X$ is a KT
    sequence $X_0 \leq \cdots \leq X_{n-1}$ and
    $C$ is a Kleene sequence $(C_i, C_{i+1}, \dots, C_{n-1})$ in $(\Set/[0, 1])_S$ ($C$ is
    empty if $n=i$).}
  \Initially{$(X; C)=(\mathbf{0} \leq \Fmdp (\mathbf{0}); ())$}  \Repeat{any return
    value is obtained}{ \textbf{Valid}
    If $X_{j+1} \leq X_j$ for some $j < n-1$, return `True' with the conclusive KT sequence $X$. \\
    \textbf{Unfold} If $X_{n-1} \leq  d_{\iota,\lambda}$
    (i.e.~$X_{n-1}(s_\iota) \leq \lambda$), let
    $(X; C) \coloneqq (X_0 \leq \cdots \leq X_{n-1} \leq \mathbf{1}; ())$. \\
    \textbf{Induction} If some $k \geq 2$ and $x\in L$ satisfy
    $X_{k} \not \leq x$ and $\Fmdp(X_{k-1} \land x) \leq x$,
    let $(X; C) \coloneqq (X[X_j := X_j \land x]_{2 \leq j \leq k}, C)$. \\
    \textbf{Candidate} If $C=()$ and $X_{n-1} \not \leq d_{\iota,\lambda}$
    (i.e.~$X_{n-1}(s_\iota) > \lambda$),
    choose $x: S \to [0, 1]$ satisfying $x \leq X_{n-1}$ and $x \not \leq  d_{\iota,\lambda}$, and
    let $(X; C) \coloneqq (X; (x))$. \\
    \textbf{Model}
    If $C_1$ is defined, return `False' with the conclusive Kleene sequence $(\mathbf{0}, C_1, \dots, C_{n-1})$. \\
    \textbf{Decide} If $C_i \leq \Fmdp X_{i-1}$ (i.e.~for all
    $s \in \alpha$, there exists $a \in \Act$ such that
    $C_i(s) \leq \sum_{s' \in S} X_{i-1}s' \cdot \delta(s)(a)(s')$),
    choose $x: S \to [0, 1]$ satisfying $x \leq X_{i-1}$ and $C_i \leq Fx$,
    and let $(X; C) \coloneqq (X; (x, C_i, \dots, C_{n-1}))$. \\
    \textbf{Conflict} If $C_i \not \leq \Fmdp X_{i-1}$ (i.e.~there
    exists $s \in \alpha$ such that
    $C_i(s) > \sum_{s' \in S} X_{i-1}s' \cdot \delta(s)(a)(s')$ for
    all $a \in \Act$), choose $x: S \to [0, 1]$ satisfying
    $C_i \not \leq x$ and $\Fmdp (X_{i-1} \land x) \leq x$, and let
    $(X; C) \coloneqq (X[X_j := X_j \land x]_{2 \leq j \leq i},
    (C_{i+1}, \dots, C_{n-1}))$. }
  \caption{LT-PDR Algorithm \IBMDP{}  for MDPs}
  \label{alg:mdp}
\end{algorithm}

\section{LT-PDR Algorithm \MRM{} for MRMs in \S{}\ref{sec:LTPDRsForMRM}} \label{ap:mrm}

See Algorithm \ref{alg:mrm}.
It solves the SP for MRMs (Def.~\ref{def:SPMRM}).

\begin{algorithm}[h]
  \SetKwInOut{Input}{Input} \SetKwInOut{Output}{Output}
  \SetKwInOut{Initially}{Initially}
  \Input{$(d_{\iota,\lambda}, \delta, d_\alpha)$}
  \Output{`True' with a conclusive KT sequence, or `False' with a conclusive Kleene sequence}
  \KwData{$(X; C)$ where $X$ is a KT
    sequence $X_0 \leq \cdots \leq X_{n-1}$
    and $C$ is a Kleene sequence $(C_i, C_{i+1}, \dots, C_{n-1})$
     in $(\Set/[0, \infty])_S$ ($C$ is empty if $n=i$).}
  \Initially{$(X; C)=(\boldsymbol{0} \leq \Fmrm(\boldsymbol{0});
    ())$} \Repeat{any return value is obtained}{ \textbf{Valid}
    If $X_{j+1} \leq X_{j}$ for some $j < n-1$, return `True' with the conclusive KT sequence $X$. \\
    \textbf{Unfold} If $X_{n-1} \leq d_{\iota,\lambda}$
    (i.e.~$X_{n-1}(s_\iota) \leq \lambda$),
    let $(X; C) \coloneqq (X_0 \leq \cdots \leq X_{n-1} \leq \boldsymbol{\infty}; ())$. \\
    \textbf{Induction} If some $k \geq 2$ and $x\in L$ satisfy
    $X_{k-1} \not \leq x$ and $\Fmrm(X_{k} \land x) \leq x$,
    let $(X; C) \coloneqq (X[X_j := X_j \land x]_{2 \leq j \leq k}, C)$. \\
    \textbf{Candidate} If $C=()$ and $X_{n-1} \not \leq d_{\iota,\lambda}$
    (i.e.~$X_{n-1}(s_\iota) > \lambda$),
    choose $x: S \to [0, \infty]$ satisfying $x \leq X_{n-1}$ and $x \not \leq d_{\iota, \lambda}$, and
    let $(X; C) \coloneqq (X; (x))$. \\
    \textbf{Model}
    If $C_1$ is defined, return `False' with the conclusive Kleene sequence $(\mathbf{0}, C_1, \dots, C_{n-1})$. \\
    \textbf{Decide} If $C_i \leq \Fmrm X_{i-1}$
    (i.e.~$C_i(s)=0$ for all $s \not \in \alpha$,
    and $C_i(s) \leq \sum_{s' \in S} \sum_{c \in \mathbb{N}} (c+X_{i-1}s') \cdot
    \delta(s)(c, s')$ for all $s \in \alpha$),
    choose $x: S \to [0, \infty]$ satisfying
    $x \leq X_{i-1}$ and $C_i \leq \Fmrm x$,
    then let $(X; C) \coloneqq (X; (x, C_i, \dots, C_{n-1}))$. \\
    \textbf{Conflict} If $C_i \not \leq \Fmrm X_{i-1}$
    (i.e.~some $s \not \in \alpha$ satisfies $C_i(s) \neq 0$,
    or some $s \in \alpha$ satisfies
    $C_i(s) > \sum_{s' \in S} \sum_{c \in \mathbb{N}} (c+X_{i-1}s') \cdot \delta(s)(c, s')$),
    choose $x: S \to [0, \infty]$ satisfying
    $C_i \not \leq x$ and $\Fmrm (X_{i-1} \land x) \leq x$, and let
    $(X; C) \coloneqq (X[X_j := X_j \land x]_{2 \leq j \leq i},
    (C_{i+1}, \dots, C_{n-1}))$.}
  \caption{LT-PDR algorithm \MRM{} for a Markov reward model}
  \label{alg:mrm}
\end{algorithm}

\section{No Adjunction in \S\ref{sec:LTPDRsForMDP}} \label{appendix:profwedge}
In general, a component  $\dot{\Gmdp }_S: (\Set/[0, 1])_S \to (\Set/[0, 1])_{(\D S + 1)^\Act}$ of
of the lifting $\dot{\Gmdp}$ does
not have a left adjoint since $\dot{\Gmdp }_S$ may not preserve $\land$.
It follows from the following calculation ($f, g: S \to [0, 1]$ and $d: \Act \to (\D S+1)$):
\begin{align*}
  &(\dot{\Gmdp }(f \land g))(d)  \\
  &= (\tau \circ (\mathcal{D}(f \land g)+1)^\Act)(d) \\
  &= \min\{\sum_{r \in [0, 1]} r \cdot \sum_{\min(fs, gs) = r} (da)(s) \mid a \in \Act, da \neq \ast\}, \\
  &(\dot{\Gmdp } f \land \dot{\Gmdp } g)(d) \\
  &= \min\{\sum_{r \in [0, 1]} r \cdot \sum_{fs = r} (da)(s),
  \sum_{r \in [0, 1]} r \cdot \sum_{gs=r} (da)(s) \mid a \in \Act, da \neq \ast\}
\end{align*}
The above two do not coincide in general.
\qed

\section{Heuristics for \IBMDP{} in \S{}\ref{sec:implEval}} \label{ap:heuristics_mdp}
The algorithm \IBMDP{} in Alg.~\ref{alg:mdp} (except for \textbf{Induction}) is determined by
heuristics, i.e.~the way of choosing $x: S \to [0, 1]$ in \textbf{Candidate}, \textbf{Decide}, and \textbf{Conflict}.
The following is the heuristics of \IBMDP{} used in \S{}\ref{sec:implEval}.
 We use a symbolic free variable $\epsilon$ for a positive margin,
 and define
 $a+\epsilon\leq b$ by $a < b$
 and $a < b+\epsilon$ by $a \leq b$
 for each $a, b \in [0, 1]$.

  (\textbf{Candidate}):
  If $C=()$ and $X_{n-1}(s_\iota) > \lambda$, let $(X;C)\coloneqq (X;(x))$
  where $x: S \to [0, 1]$ maps $s_\iota$ to $\lambda+\epsilon$ and others to $0$.

  (\textbf{Decide}):
  If $C_i \leq \Fmdp (X_{i-1})$ (i.e.~for all
    $s \in \alpha$, there exists $a_s \in \Act$ such that
    $C_i(s) \leq \sum_{s' \in S} X_{i-1}s' \cdot \delta(s)(a_s)(s')$),
  let $(X;C)\coloneqq (X;(x,C_i,\dots,C_{n-1}))$
  where $x: S \to [0, 1]$ is defined as follows.
    Let
    $a_s \in \Act$ be an action for $s \in \alpha$ satisfying $C_i(s) \leq \sum_{s' \in S} X_{i-1}s' \cdot \delta(s)(a_s)(s')$,
    and $V$ be the set $\{s' \in S \mid \delta(s)(a_s)(s') \neq 0 \text{ for some } s \in \supp{C_i} \cap \alpha\}$.
    Then we define $x$ as
    \begin{align*}
      x(s) \coloneqq
      \begin{cases}
        0 &\text{if }s \not \in V \\
        x_s &\text{if }s \in V \text{ and }x_s=X_{i-1}s \\
        x_s+\epsilon &\text{otherwise}
      \end{cases}
    \end{align*}
    where $x_s$ is determined by
    solving the following linear program:
    find $(x_s)_{s \in V}$ that minimize $\Sigma_{s \in V} (2-X_{i-1}s)x_s$ subject to $\{v_s \leq \Sigma_{s' \in S} x_{s'} \cdot \delta(s)(a_s)(s') \mid s \in \supp{C_i} \cap \alpha, C_i(s)=v_s \text{ or }C_i(s)=v_s+\epsilon \text{ for some }v_s\in[0, 1]\} \cup \{0 \leq x_s \leq X_{i-1}s \mid s \in V\}$.

  (\textbf{Conflict}):
    If $C_i > \Fmdp (X_{i-1})$
    (i.e.~there
    exists $s \in \alpha$ such that
    $C_i(s) > \sum_{s' \in S} X_{i-1}s' \cdot \delta(s)(a)(s')$ for
    all $a \in \Act$),
    $A := \{s \in \alpha \mid C_i(s) > \sum_{s' \in S} X_{i-1}s' \cdot \delta(s)(a)(s') \text{ for all }a \in \text{Act}\}$ is not empty.
    Then let $(X;C)\coloneqq (X[X_j\coloneqq X_j \land x]_{2 \leq j \leq i};(C_{i+1},\dots,C_{n-1}))$
    where
    $x: S \to [0, 1]$ maps
    $s \not \in A$ to $1$,
    $s \in A$ with $C_i(s)=v+\epsilon$ to $v$,
    and others to $\Fmdp X_{i-1}(s)$.

 Note that $C_i(s)$ is always $v \in [0, 1]$ or $v+\epsilon$ for some $v \in [0, 1)$ by rules defined above.
 When applying \textbf{Conflict},
 each values of $\epsilon$ in the Kleene sequence $C$ can be implicitly determined
 as small enough ones so that all conditions in rules (e.g.~$C_i\leq X_i$ and $C_i \leq \Fmdp (X_{i-1})$) hold.
 By this fact the heuristics above is valid for Alg.~\ref{alg:mdp}.
 The heuristics of \MRM{} in \S{}\ref{sec:implEval} is similarly designed.

\section{Full Experiment Results for \IBMDP{}} \label{appendix:fullMDP}
See Table~\ref{table:expMDPextended}.
  \begin{table}[htp]
  \caption{Results with \IBMDP. Comparison is against PrIC3~\cite{BatzJKKMS20} with four
different interpolation generalization methods (none, linear, polynomial, hybrid). The benchmarks are from~\cite{BatzJKKMS20}.
$|S|$ is the number of states of the benchmark MDP. ``GT pr.'' is for the \emph{ground truth probability}, that is the reachability probability $\mathit{Pr}^{\mathit{max}}(s_\iota \models \diamond (S \setminus\alpha))$ computed outside the solvers under experiments. The solvers were asked whether the GT pr.\ (which they do not know) is $\le \lambda$ or not; they all answered correctly. The last five columns show the average execution time in seconds.  -- is for ``did not finish,'' for out of memory or timeout (600 sec.)
}
    \label{table:expMDPextended}
    \centering
    \begin{tabular}{ccccccccc}
      \toprule
  Benchmark & $|S|$ & $\mathit{Pr}^{\mathit{max}}(s_\iota \models \diamond (S\setminus\alpha))$ & $\lambda$ & \IBMDP{} & w/o & lin & pol & hyb  \\
      \midrule
    \multirow{2}{*}{Grid} & \multirow{2}{*}{$10^2$} & \multirow{2}{*}{$1.2E^{-3}$} & 0.3 &0.31 & 1.31 & 19.34 & -- & --  \\
      &   &   & 0.2 &0.48 & 1.75 & 24.62 & -- & --\\
      \midrule
    \multirow{2}{*}{Grid} & \multirow{2}{*}{$10^3$} & \multirow{2}{*}{$4.4E^{-10}$} & 0.3 &112.29 & -- & -- & -- & -- \\
      &   &   & 0.2 &136.46 & -- & -- & -- & --\\
      \midrule
    \multirow{3}{*}{BRP} & \multirow{3}{*}{$10^3$} & \multirow{3}{*}{0.035} & 0.1 &-- & -- & -- & -- & -- \\
      &   &   & 0.01  &18.52 & 56.55 & 594.89 & -- & 722.38  \\
      &   &   & 0.005 &1.36  & 11.68 & 238.09 & -- & --   \\
      \midrule
    \multirow{4}{*}{ZeroConf} & \multirow{4}{*}{$10^4$} & \multirow{4}{*}{0.5} & 0.9 &-- & -- & -- & 0.58 & 0.51 \\
      &   &   & 0.75 & -- & -- & -- & 0.55 & 0.46 \\
      &   &   & 0.52 & -- & -- & -- & 0.48 & 0.46 \\
      &   &   & 0.45 & 0.014 & $<$0.1 & $<$0.1 & $<$0.1 & $<$0.1 \\
      \midrule
    \multirow{4}{*}{Chain} & \multirow{4}{*}{$10^3$} & \multirow{4}{*}{0.394} & 0.9 & -- & 72.37 & -- & 0.91 & 0.70 \\
      &   &   & 0.4 & -- & 80.83 & -- & 0.93 & -- \\
      &   &   & 0.35 & 177.12 & 115.98 & -- & -- & --  \\
      &   &   & 0.3 & 88.27 & 66.89 & 557.68 & -- & -- \\
      \midrule
    \multirow{4}{*}{Chain} & \multirow{4}{*}{$10^4$} & \multirow{4}{*}{0.394} & 0.9 & -- & -- & -- & 0.86 & 0.63 \\
      &   &   & 0.48 & -- & -- & -- & 0.84 & -- \\
      &   &   & 0.4 & -- & -- & -- & 0.84 & -- \\
      &   &   & 0.3 & -- & -- & -- & -- & -- \\
      \midrule
    \multirow{2}{*}{Chain} & \multirow{2}{*}{$10^{12}$} & \multirow{2}{*}{0.394} & 0.9 & -- & -- & -- & 0.91 & -- \\
      &   &   & 0.4 & -- & -- & -- & 0.89 & -- \\
      \midrule
    \multirow{4}{*}{DoubleChain} & \multirow{4}{*}{$10^3$} & \multirow{4}{*}{0.215} & 0.9 & -- & -- & -- & 1.83 & 1.99 \\
      &   &   & 0.3 & -- & -- & -- & 1.88 & 1.96 \\
      &   &   & 0.216 & -- & -- & -- & 139.76 & -- \\
      &   &   & 0.15 & 7.46 & -- & -- & -- & -- \\
      \midrule
    \multirow{3}{*}{DoubleChain} & \multirow{3}{*}{$10^4$} & \multirow{3}{*}{0.22} & 0.9 & -- & -- & -- & 1.83 & 2.47 \\
      &   &   & 0.3 & -- & -- & -- & 2.11 & 2.00\\
      &   &   & 0.24 & -- & -- & -- & 2.01 & -- \\
      \bottomrule
    \end{tabular}
  \end{table}

\section{Omitted Proofs}
\subsection{Proof of Cor.~\ref{cor:kt_kleene}}
\begin{proof}
  1) easy.
  2) By Thm.~\ref{thm:kt_cc}, we have the following.
  \begin{align*}
    &\mu F \not \leq \alpha \\
    &\Leftrightarrow \text{there exists $n \in \mathbb{N}$ such that } F^n \bot \not \leq \alpha \\
    &\Leftrightarrow  \text{there exists $n \in \mathbb{N}$ and $x \in L$ such that } x \leq F^n \bot \text{ and } x \not \leq \alpha.
  \end{align*}
  \qed
\end{proof}
\subsection{Proof of Thm.~\ref{thm:safe_witness}}
\begin{proof}
  Since $L$ is a complete lattice, we consider a monotone function
  $\sup: [\omega, L] \to L$ mapping $X$ to
  $\bigvee_{i \in \omega} X_i$, which has the upper (i.e.~right)
  adjoint $\Delta: L \to [\omega, L]$.
  \begin{displaymath}
    \xymatrix@C=2cm{
      [\omega, L] \lloop{F^\#} \ar@<1.2ex>[r]^-{\sup} &L \ar@<1.2ex>[l]^-{\Delta}_-\bot \rloop{F} \\
      \Pref{F^\#} \ar@{^(->}[u] \ar@<1.2ex>[r]^-{\Pref{\sup}} &\Pref{F} \ar@<1.2ex>[l]^-{\Pref{\Delta}}_-\bot \ar@{^(->}[u]
    }
  \end{displaymath}
  Since $F$ is $\omega$-continuous, the monotone function $F^\#$ is a
  lifting of $F$ along $\sup$, that is, $\sup\circ F^\#=F\circ \sup$
  holds.
  Now one can easily check that we can restrict $\sup$ and $\Delta$ to
  functions between $\Pref{F^\#}$ and $\Pref{F}$, and a general result
  in category theory \cite[Thm 2.14]{HermidaJ98} tells us that the
  restrictions (denoted as $\Pref{\sup}$ and $\Pref{\Delta}$ in the
  above diagram) again form a Galois connection. Note that the initial
  chain of $F^\#$ is $\mu F^\#$, and is mapped to $\mu F$ by
  $\Pref{\sup}(\mu F^\#)=\mu F$. We will use this fact in some later proofs.

  Then $\Pref{\Delta}$ maps a KT witness to a KT$^\omega$ witness,
  and $\Pref{\sup}$ maps a KT$^\omega$ witness to a KT witness.
  \qed
\end{proof}

\subsection{Proof of Thm.~\ref{thm:KTseqCPO}}
\begin{proof}
  Assume $X^0 \preceq X^1 \preceq \cdots$ is an $\omega$-chain of KT sequences augmented with KT$^\omega$ witnesses.
  Then the suprema of this chain exist: its $j$-th element is the infimum of $\{X_j^i \mid X_j^i \text{ is defined}\}$ in $L$.
  The suprema
  compose $F^\#_n$ or $F^\#$-algebra and each element is less than or equal to $\alpha$.
  \qed
\end{proof}

\subsection{Proof of Thm.~\ref{thm:positive_sound_terminate}}
\begin{proof}
  (sound) easy by Cor.~\ref{cor:kt_kleene}, Thm.~\ref{thm:safe_witness}, and Prop.~\ref{prop:safe}.
  (weakly terminating)
  If $\mu F \leq \alpha$ then the algorithm weakly terminates by the following procedure (skip \textbf{Induction} when we cannot apply the rule):
  $(\bot \leq F\bot) \xmapsto{\mathbf{Unfold}}
  (\bot \leq F\bot \leq \top) \xmapsto{\mathbf{Induction}}
  (\bot \leq F\bot \leq \mu F) \xmapsto{\mathbf{Unfold}}
  (\bot \leq F\bot \leq \mu F \leq \top) \xmapsto{\mathbf{Induction}}
  (\bot \leq F\bot \leq \mu F \leq \mu F) \xmapsto{\mathbf{Valid}}
  \text{`True'}
  $.
  \qed
\end{proof}

\subsection{Proof of Lem.~\ref{lem:kt_order}}
\begin{proof}
  When we cannot apply both \textbf{Valid} and \textbf{Unfold},
  we can apply \textbf{Induction} by choosing $\mu F$ as $x$.
  \qed
\end{proof}

\subsection{Proof of Thm.~\ref{thm:posi_kt}}
\begin{proof}
  (non-termination)
  Since $L$ is well-founded, a non-terminating run $X^0 \preceq X^1 \preceq \cdots$ infinitely extends the length of KT sequences.
  Therefore, by Thm.~\ref{thm:KTseqCPO},
  the supremum of the $\omega$-chain becomes a KT$^\omega$ witness.

  (strong termination)
  Assume there is a run of positive LT-PDR which does not terminate.
  Let $X^i$ be the $i$-step KT sequence and $X^i_j\coloneqq\top$ when $|X^i| \leq j$.
  The $\omega$-chain $\bigwedge_{i \in \mathbb{N}} X^i$ is under $\alpha$ so it converges in some index $j$: $\bigwedge_{i \in \mathbb{N}}X^i_j = \bigwedge_{i \in \mathbb{N}}X^i_{j+1}$.

  We further assume that there is no $i \in \mathbb{N}$ such that $X^i_j = X^i_{j+1}$.
  Then for each $i \in \mathbb{N}$, there exists $i'$ such that $X^i_j \geq X^{i'}_{j+1}$ since
  $\bigwedge_{i \in \mathbb{N}}X^i_j = \bigwedge_{i \in \mathbb{N}}X^i_{j+1}$.
  Now $X^{i'}_j \neq X^{i'}_{j+1}$ so $X^{i'}_j < X^{i'}_{j+1}$ holds.
  Applying it repeatedly, we have $X^0_j > X^{0'}_j > X^{0''}_j > \dots$.
  This contradicts well-foundedness.
  \qed
\end{proof}

\subsection{Proof of Thm.~\ref{thm:negative}}
\begin{proof}
  1) easy by Cor.~\ref{cor:kt_kleene} and Prop.~\ref{prop:unsafe}.
  2) By Thm.~\ref{thm:kt_cc}, there exists $n \in \mathbb{N}$ such that $F^n \bot \not \leq \alpha$. Then negative LT-PDR terminates when we choose $x$ in \textbf{Candidate} and \textbf{Decide} so as to get the conclusive Kleene sequence $(\bot, F\bot, \dots, F^n \bot)$.
  \qed
\end{proof}

\subsection{Proof of Prop.~\ref{prop:C_X}}
\begin{proof}
  1) If $C_i \not \leq X_i$ then $C_i \not \leq F^i \bot$ by Lem.~\ref{lem:x_init}.
  Lem.~\ref{lem:C_X} (not \ref{item:fibot} $\Rightarrow$ not \ref{item:possible}) concludes the proof.

  2) Considering $j=1$, Lem.~\ref{lem:C_X} (not \ref{item:fjx} $\Rightarrow$ not \ref{item:possible}) concludes the proof.

  3) By Lem.~\ref{lem:x_init} and the KT sequence $(X_0 \leq \dots \leq X_{n-1})$,
  we have $F^{n-2}\bot \leq X_{n-2} \leq \alpha$.
  Letting $i=n-2$, Lem.~\ref{lem:C_X} (not \ref{item:fibot} $\Rightarrow$ not \ref{item:possible}) concludes the proof.
  \qed
\end{proof}

\subsection{Proof of Lem.~\ref{lem:x_init}}
\begin{proof}
  In the proof of Thm.~\ref{thm:safe_witness},
  we showed $\mu F^\#$ is the initial chain of $F$.
  Therefore, each prefixed point of $F^\#$ is greater than or equal to the initial chain of $F$.
  This fact leads to the over-approximation of KT sequences.
  \qed
\end{proof}

\subsection{Proof of Lem.~\ref{lem:C_X}}
\begin{proof}
  ($1 \Rightarrow 2$):
  $C_i \leq FC_{i-1} \leq \cdots \leq F^i C_0 = F^i \bot$.
  ($2 \Rightarrow 1$): It is true since
  $(\bot, F\bot, \dots, F^{i-1}\bot, C_i, \dots, C_{n-1})$ is a
  conclusive Kleene sequence.  ($2 \Rightarrow 3$):
  $F^i \bot \leq F^j X_{i-j}$ by Lem.~\ref{lem:x_init}.  \qed
\end{proof}
\subsection{Proof of Lem.~\ref{lem:config}} \label{ap:config}
\begin{proof}
  Preservation of Kleene sequences is easily proved.
  We prove the preservation of KT sequences
  by checking each condition in Def.~\ref{def:kt_sequence}.
  \begin{enumerate}
  \item  The initial $X$ satisfies $X_{n-2} \leq \alpha$ because
    $\bot \leq \alpha$.  Rules except for \textbf{Unfold} cannot
    increase $X$, especially $X_{n-2}$, and \textbf{Unfold} also
    preserves $X_{n-2} \leq \alpha$.
  \item
    The LT-PDR algorithm starts from $X=(\bot \leq F\bot) \in [2, L]$
    composing $F^\#_2$-algebra.  All rules which update $X$ are the
    following:
    \begin{itemize}
    \item (\textbf{Unfold}): For each $n, m$ with
      $n \leq m \leq \omega$, let $a$ denote the functor from
      $[m, L] \to [n, L]$ which shortens sequences by cutting large
      elements.  Each $a$ has a right adjoint $r$ which appends a
      sequence by $\top \in L$.

      The rule sends $X \in [n, L]$ to $rX \in [n+1, L]$ by $r$ and we
      show $r$ sends $F^\#_n$-algebra to $F^\#_{n+1}$-algebra.  As the
      same discussion in the proof of Thm.~\ref{thm:safe_witness}, Thm 2.14
      in \cite{HermidaJ98} yields the following since $F^\#_{n+1}$ is
      a lifting of $F^\#_n$ along $a$:
      \begin{displaymath}
        \xymatrix{
          [n, L] \lloop{F^\#_{n}} \ar@<-1.2ex>[r]_-{r} &[n+1, L] \ar@<-1.2ex>[l]_-{a}^-\bot \rloop{F^\#_{n+1}}
        }
        \text{ gives }
        \xymatrix{
          \Pref{F^\#_{n}} \ar@<-1.2ex>[r]_-{\Pref{r}} &\Pref{F^\#_{n+1}} \ar@<-1.2ex>[l]_-{\Pref{a}}^-\bot
        }.
      \end{displaymath}
      Thus, $r$ preserves algebras.

    \item (\textbf{Induction, Conflict}): These two rules preserve
      prefixed points of $F^\#_n$ because $F(X_{k-1} \land x) \leq x$ iff
      $F^\#_n(r \Delta x \land X) \leq r\Delta x$ ($\Delta: L \to [k+1, L]$ and $r: [k+1, L] \to [n, L]$) by the following:
      \begin{displaymath}
        \infer=[^{a \dashv r: [k+1, L] \to [n, L]}]{F^\#_n(X \land r \Delta x) \leq r \Delta x}{\infer=[^{\sup \dashv \Delta: L \to [k+1, L]}]{a F^\#_n (X \land r \Delta x) \leq \Delta x}{F(X_{k-1} \land x) = \sup a F^\#_n (X \land r \Delta x) \leq x}}
      \end{displaymath}
    \end{itemize}
  \end{enumerate}
  \qed
\end{proof}
\subsection{Proof of Prop.~\ref{prop:term}}
\begin{proof}
  Note that all rules in LT-PDR change the current data $(X; C)$.
  Since we have well-foundedness and the length of sequences in the data is always finite,
  \textbf{Unfold} or \textbf{Model} will be applied within finite steps.

  When $\mu F \not \leq \alpha$ is true,
  there exists a conclusive Kleene sequence by Cor.~\ref{cor:kt_kleene} and Prop.~\ref{prop:unsafe}.
  Letting $n$ be the length of the sequence,
  by Prop.~\ref{prop:C_X}.\ref{item:n-step}, there is no KT sequence with length $n+1$.
  Thus the algorithm will terminate by \textbf{Model} in finite steps.

  When $\mu F \leq \alpha$ is true and $(L, \leq)$ has no strictly increasing $\omega$-chain bounded by $\alpha$,
  we cannot apply \textbf{Unfold} infinitely.
  Thus the algorithm will terminate by \textbf{Unfold} within finite steps.
  \qed
\end{proof}

\subsection{Proof of Prop.~\ref{prop:safe_prob_fd}}
\begin{proof}
  \begin{align*}
    &\iota \;\leq\; \nu x.\, \alpha \land \delta^*\dot{G}x \\
    &\text{iff there exists a coalgebra } x \leq \alpha \land \delta^* \dot{G}x \text{ in } \mathbb{E}_S \text{ satisfying } \iota \leq x \\
    &\text{iff there exists } x \text{ in } \mathbb{E}_S \text{ satisfying } \iota \leq x \leq \alpha \text{ and } x \leq \delta^* \dot{G}x  \\
    &\text{iff there exists } x \text{ in } \mathbb{E}_S \text{ satisfying } \iota \leq x \leq \alpha \text{ and } \dot{F}\delta_*x \leq x  \\
    &\text{iff there exists an algebra } \iota \lor \dot{F}\delta_* x \leq x \text{ in } \mathbb{E}_S \text{ satisfying } x \leq \alpha \\
    &\text{iff } \mu x.\,\iota \lor \dot{F}\delta_*x \;\leq \;\alpha.
  \end{align*}
      \qed
\end{proof}

\section{Haskell Source Code for LT-PDR}
\label{appendix:code}
The following is our Haskell implementation \lstinline{ltPDR}.
{
\scriptsize\begin{verbatim}
class (Show a) => CLat a where
  type Info a            -- auxiliary information
  leq :: a -> a -> IO (Bool, Info a)
  bot :: a -> a          -- include dummy argument
  top :: a -> a
  meet :: a -> a -> a

type KTSeq a = [a]         -- [X_{n-1}, ..., X_1=f bot]
type KleeneSeq a = Stack a -- Stack (n-i) [C_i, ..., C_{n-1}]
newtype PDRConfig a = KTKl (KTSeq a, KleeneSeq a) deriving (Show)
data PDRAnswer a = Valid (KTSeq a) | InValid (KleeneSeq a) deriving (Show)
data (CLat a) => Heuristics a = Heuristics { f_candidate  :: a -> a -> Info a -> IO a,
                                             f_decide   :: a -> a -> (a -> a) -> Info a -> IO a,
                                             f_conflict :: a -> a -> (a -> a) -> Info a -> IO a }

-- check whether mu F <= alpha
ltPDR :: forall a. CLat a => Heuristics a -> (a -> a) -> a -> IO (PDRAnswer a)
ltPDR heuristics f alpha =
  let init = KTKl ([f $ bot alpha], stackNew) in
    loop init
  where
    loop :: PDRConfig a -> IO (PDRAnswer a)
    loop (KTKl (xs, cs)) = do
      rst <- sequence $ [fst <$> leq (xs !! i) (xs !! (i+1)) | i <- [0..(length xs - 2)]]
      if or rst
        then return $ Valid xs
        else if length xs == naturalToInt (stackSize cs)
          then return $ InValid cs
          else do
            (result1, solver1) <- leq (head xs) alpha
            if result1
              then loop $ KTKl (top alpha:xs, stackNew)
              else
                case stackPop cs of
                  Nothing -> do
                    x <- f_candidate heuristics (head xs) alpha solver1
                    loop $ KTKl (xs, stackPush cs x)
                  Just (cs', ci) ->
                    let sizeOfcs = naturalToInt $ stackSize cs in
                    let xi1 = xs !! sizeOfcs in do
                      (result2, solver2) <- leq ci (f xi1)
                      if result2
                        then do
                          x <- f_decide heuristics xi1 ci f solver2
                          loop $ KTKl (xs, stackPush cs x)
                        else do
                          x <- f_conflict heuristics xi1 ci f solver2
                          let sizeOfxs = length xs
                          let xs' = zipWith (h sizeOfcs sizeOfxs x) xs [0..]
                          loop $ KTKl (xs', cs')
    h :: Int -> Int -> a -> a -> Int -> a
    h sizeOfcs sizeOfxs x x' i = if i == sizeOfxs - 1 || i < sizeOfcs - 1 then x' else meet x x'
\end{verbatim}
}
\else
\fi

\end{document}